\pdfoutput=1
\documentclass[pdftex,pra,aps,twocolumn,nofootinbib]{revtex4-1}

\usepackage[T1]{fontenc}
\fontencoding{T1} 
\usepackage[utf8]{inputenc}

\usepackage[dvipsnames]{xcolor}
\usepackage{tikz}
\usepackage{verbatim}
\usetikzlibrary{arrows,shapes}
\usetikzlibrary{decorations.pathmorphing}
\usetikzlibrary{decorations.markings}

\usepackage{amsfonts}
\usepackage{amsmath,amssymb,amsthm}
\usepackage{mathtools}
\usepackage{graphicx}
\usepackage{fancyhdr}
\usepackage{tcolorbox}
\usepackage{bbm}
\usepackage[T1]{fontenc}
\usepackage{url}

\usepackage{diagbox}

\usepackage{braket}
\usepackage{upgreek}
\usepackage{dsfont}
\usepackage{enumitem}

\usepackage[plain]{fancyref} 
\usepackage{algorithm}
\usepackage{algpseudocode}

\usepackage{hyperref}

\theoremstyle{plain}

\newtheorem{prop}{Proposition}

\newtheorem{observation}{Observation}
\theoremstyle{remark}

\newcommand*{\fancyrefthmlabelprefix}{thm}
\newcommand*{\fancyreflemlabelprefix}{lem}
\newcommand*{\fancyrefcorlabelprefix}{cor}
\newcommand*{\fancyrefdefilabelprefix}{defi}
\frefformat{plain}{\fancyreflemlabelprefix}{lemma\fancyrefdefaultspacing#1}
\Frefformat{plain}{\fancyreflemlabelprefix}{Lemma\fancyrefdefaultspacing#1}
\frefformat{plain}{\fancyrefthmlabelprefix}{theorem\fancyrefdefaultspacing#1}
\Frefformat{plain}{\fancyrefthmlabelprefix}{Theorem\fancyrefdefaultspacing#1}
\frefformat{plain}{\fancyrefcorlabelprefix}{corollary\fancyrefdefaultspacing#1}
\Frefformat{plain}{\fancyrefcorlabelprefix}{Corollary\fancyrefdefaultspacing#1}
\frefformat{plain}{\fancyrefdefilabelprefix}{definition\fancyrefdefaultspacing#1}
\Frefformat{plain}{\fancyrefdefilabelprefix}{Definition\fancyrefdefaultspacing#1}
\newcommand*{\fancyrefalglabelprefix}{alg}
\newcommand*{\frefalgname}{algorithm}
\newcommand*{\Frefalgname}{Algorithm}
\frefformat{plain}{\fancyrefalglabelprefix}{%
	\frefalgname\fancyrefdefaultspacing#1%
}%
\Frefformat{plain}{\fancyrefalglabelprefix}{%
	\Frefalgname\fancyrefdefaultspacing#1%
}%

\def\beq{\begin{equation}}
\def\eeq{\end{equation}}
\def\bq{\begin{quote}}
	\def\eq{\end{quote}}
\def\ben{\begin{enumerate}}
	\def\een{\end{enumerate}}
\def\bit{\begin{itemize}}
	\def\eit{\end{itemize}}

\def\r|{\right|}

\newcommand{\tr}{\operatorname{tr}}
\newcommand{\ot}{\otimes}
\newcommand{\A}{\widetilde{A}}

\newcommand\be{\begin{equation}}
\newcommand\ee{\end{equation}}
\newcommand{\<}{\langle}
\renewcommand{\>}{\rangle}

\newcommand{\vc}{\vcentcolon=}

\tikzstyle{vertex}=[circle,edge=black!40,fill=blue!35,minimum size=35pt,inner sep=1pt]
\tikzstyle{vertex3}=[rectangle,fill=none,minimum size=0pt,inner sep=1pt]
\tikzstyle{vertex2}=[circle,fill=blue!25,minimum size=0pt,inner sep=1pt]
\tikzstyle{edge} = [draw,thick,->]
\tikzstyle{edge2} = [draw,thick,dashed]
\tikzstyle{edge3} = [draw,thick,->,dashed]
\tikzstyle{edge4} = [draw,double,->, dashed]
\tikzstyle{edge5} = [draw,thick,-,decorate, decoration={snake}]
\tikzstyle{edge6} = [draw,thick,->,blue]
\tikzstyle{weight} = [font=\small]
\tikzstyle{selected edge} = [draw,line width=5pt,-,red!50]
\tikzstyle{ignored edge} = [draw,line width=5pt,-,black!20]

\tikzstyle{fillc0} = [rectangle, rounded corners, edge=black!40,minimum width=90, minimum height=40, fill=Aquamarine!20]
\tikzstyle{fillc1} = [rectangle, rounded corners, edge=black!40,minimum width=90, minimum height=40, fill=GreenYellow!40]

\definecolor{dg}{rgb}{0,.5,0}

\begin{document}
\title{Two instances of random access code in the quantum regime}

\author{Nitica Sakharwade$^{1,2,3}$, Micha{\l} Studzi\'nski$^{4,}\footnote{corresponding author, email: studzinski.m.g@gmail.com}$, Micha{\l} Eckstein$^{5}$ and Pawe{\l} Horodecki$^{1,6}$}
	\affiliation{
	$^1$ International Centre for Theory of Quantum Technologies, University of Gda{\'n}sk, Jana Ba{\.z}y{\'n}skiego 1A, 80-309 Gda{\'n}sk, Poland\\
 $^2$ Perimeter Institute for Theoretical Physics, Waterloo, Canada \\
 $^3$ Department of Physics and Astronomy, University of Waterloo, Waterloo, Canada\\
	$^4$ Institute of Theoretical Physics and Astrophysics and National Quantum Information Centre in Gda{\'n}sk,
	Faculty of Mathematics, Physics and Informatics, University of Gda{\'n}sk, Wita Stwosza 57, 80-952 Gda{\'n}sk, Poland\\
	$^5$ Institute of Theoretical Physics, Faculty of Physics, Astronomy and Applied Computer Science, Jagiellonian University,
ul. {\L}ojasiewicza 11, 30–348 Krak{\'o}w, Poland\\
	$^6$ Faculty of Applied Physics and Mathematics, National Quantum Information Centre,
 Gdansk University of Technology, Gabriela Narutowicza 11/12, 80-233 Gda{\'n}sk, Poland}

\begin{abstract}
We consider two classes of quantum generalisations of Random Access Code (RAC) and study the bounds for probabilities of success for such tasks\footnote{The published version of this work can be found in \cite{sakharwade2022two}}. It provides a useful framework for the study of certain information processing tasks with constrained resources. The first class is based on a random access code with quantum inputs and output known as the No-Signalling Quantum RAC (NS-QRAC) box [A. Grudka {\it et al.}, Phys. Rev. A {\bf 92}, 052312 (2015)], where unbounded entanglement and constrained classical communication are allowed. We show that it can be seen as quantum teleportation with constrained classical communication and provide a lower quantum bound for the success probability. We consider two modifications to the NS-QRAC scenario: The first, where unbounded entanglement and constrained quantum communication is allowed and the second, where bounded entanglement and unconstrained classical communication is allowed. We find a monogamy relation for the transmission fidelities, which --- in contrast to the usual communication schemes --- involves multiple senders and a single receiver. We provide an upper bounds for the latter and a lower one for the former. The second class is based on a RAC with a quantum channel and shared entanglement [A. Tavakoli {\it et al.}, PRX Quantum \textbf{2}, 040357 (2021)]. We study the set of tasks where two inputs made of two digits of $d$-base are encoded over a qudit and a maximally entangled state. We show that such tasks can be seen as quantum dense-coding with constrained quantum communication and explicit protocols, which give lower quantum bounds for the tasks' efficiency, in dimensions $d=2,3,4$. The employed encoding utilises Gray codes.
\end{abstract}
\maketitle

\section{Introduction}

Quantum information processing protocols offer an unprecedented advantage over classical schemes \cite{NielsenChuang} providing new resources for computation, communication or cryptography \cite{QResources}. Central to the study quantum communication is the characterisation of possible protocols based on signalling and non-signalling resources and classical and quantum resources, such as in teleportation \cite{bennett_teleporting_1993} and dense-coding \cite{bennett1992communication}, that allow for some quantum advantage. 

However, in practical scenarios, one may face some constraints such as a limited amount of available resources (quantum or classical) that restrict the possible communication tasks that can be performed perfectly \cite{korzekwa2019,Hsieh20}. This is due to the Holevo bound, which says that one cannot faithfully encode more than $m$ bits of information on $m$ qubits. Quantum Random Access Codes (QRACs) \cite{RAC83,RAC} circumvent this by allowing imperfect fidelity in return for the possibility to encode more information. 

QRACs, a generalisation of classical Random Access Codes (RACs), date back to 1983 (then named `conjugate coding') \cite{RAC83}, and have many practical applications including network coding and cryptography. In QRACs one party encodes a string of bits on a quantum system sent to a second party who wishes to decode a subset of the string. QRACs have been extensively studied \cite{nayak1999optimal, ambainis2002dense, spekkens2009preparation, tavakoli2015quantum, carmeli2020quantum}. In particular, it plays a key role in formulating the principle of information causality \cite{InformationCausality}. Another well-known generalisation of RAC, called Entanglement Assisted RAC (EA-RAC) \cite{pawlowski2010entanglement}, studies encoding information over entanglement-assisted classical channels.

\begin{figure}[h]
    \centering
\begin{tikzpicture}[scale=0.20,auto,swap]

\fill[blue!80!red!20!] (0,-3) rectangle (6,3);
\draw (0,-3) rectangle (6,3);
\draw (3,0) node {\large Alice};
\fill[blue!80!red!20!] (0+15,-3) rectangle (6+15,3);
\draw (0+15,-3) rectangle (6+15,3);
\draw (3+15,0) node {\large Bob};
\draw[->] (2+7.5,-6) -- (5,-3);
\draw[->] (4+7.5,-6) -- (1+15,-3);
\fill[blue!80!red!20!] (1+7.5,-10) rectangle (5+7.5,-6);
\draw (1+7.5,-10) rectangle (5+7.5,-6);
\draw (3+7.5,-8) node {\large NS};
\draw[->] (6,0) -- (15,0);
\draw (10.5,1) node {\large $m$};
\draw[->] (1,6) -- (1,3);
\draw[->] (3,6) -- (3,3);
\draw[->] (5,6) -- (5,3);
\draw (1,7) node {\large $\rho_1$};
\draw (3,7) node {\large $\rho_2$};
\draw (5,7) node {\large $...$};
\draw[->] (3+15,6) -- (3+15,3);
\draw[->] (3+15,-3) -- (3+15,-6);
\draw (3+15,7) node {\large $c$};
\draw (3+15,-7) node {\large $\rho_c$};

\end{tikzpicture}
    \caption{Instances of RACs in the quantum regime.}
    \label{fig:intro}
\end{figure}
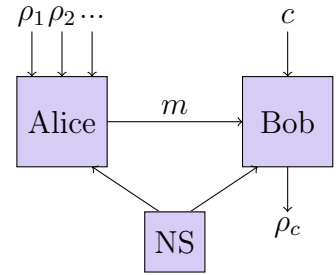

In this work we study the generalisations of RACs in the quantum regime, which go, broadly speaking, in two main directions: 1) RACs seen as constrained teleportation, where Alice wishes to teleport to Bob one of two quantum states unknown using constrained teleportation resources and 2) RACs seen as constrained dense-coding, where Alice wishes to dense code one of two classical strings unknown to Bob using constrained dense-coding resources. 

Different generalisations of RAC in the quantum regime can be understood with the help of the general setup presented in Figure \ref{fig:intro}. Alice encodes multiple states $\rho_i$ into a message $m$ with the aid of some no-signalling resources. Bob wishes to decode $\rho_c$, given his choice $c$. The following scenarios may be considered:

\begin{enumerate}[itemsep=0mm]
    \item The inputs/outputs $\rho_i,\rho_c$ are classical or quantum.
    \item The message $m$ sent is classical or quantum.
    \item The parties share no-signalling resources, such as shared randomness or entanglement.
    \item The channel has (un)constrained capacity.
    \item The no-signalling resource is (un)bounded. 
\end{enumerate}

The above options are classified in Table~\ref{Tab0}. 

\begin{table}[h]
\begin{center}
 \begin{tabular}{|c|c||c|c|c|c|}
 \hline
 & Scenario & Inputs & Channel & No-signalling & Outputs \\
 \hline
 1 & RAC & Cl & Cl & SR & Cl \\
 2 & QRAC & Cl & Q & SR & Cl \\
 3 & EA-RAC & Cl & Cl & Ent. & Cl \\
 4 & QRAC-SE & Cl & Q & Ent. & Cl \\
 5 & NS-QRAC & Q & Con. Cl & Unb NS & Q \\ 
 6 & NS-QRAC & Q & Con. Cl & Unb Ent. & Q \\ 
 7 & NS-QRAC & Q & Con. Q & Unb Ent. & Q \\  
 8 & CNS-QRAC & Q & UnC. Cl & B Ent. & Q \\  
 \hline
 \end{tabular}
 \caption{Quantum generalisations of RACs. Here Cl stands for classical, Q for quantum, Con. stands for Constrained, UnCon. stands for Unconstrained, SR stands for Shared Randomness and Ent. stands for entanglement, NS for No-Signalling resources Unb for Unbounded and B for Bounded.}
 \label{Tab0}
 \end{center} 
\end{table}
One class of quantum generalisations of RACs uses quantum resources for the transmission of classical information. Such quantum variations of RAC include three broad categories: 1) The communication channel is quantum and the parties share randomness --- Quantum Random Access Codes with Shared Randomness (QRAC-SR) \cite{ambainis2008quantum} (Row 2 of Table~(\ref{Tab0})), 2) the channel is classical and the parties share entanglement --- Entanglement Assisted Random Access Codes (EA-RAC) \cite{pawlowski2010entanglement} (Row 3 of Table~(\ref{Tab0})), and 3) the channel is quantum and the parties share entanglement \cite{tavakoli2021correlations} --- which we will refer to as Quantum Random Access Codes with Shared Entanglement (QRAC-SE) (Row 4 of Table~(\ref{Tab0})). In Section III B of \cite{tavakoli2021correlations} quantum upper bounds for the probability of success were studied for a low number of inputs. In this work we study lower quantum bounds for QRAC-SE for higher number of inputs, which we will also employ for our modification of the NS-QRAC scenario.

Another class of quantum generalisation of RAC concerns the transmission of quantum states, rather than classical bits, where the inputs itself are quantum, which has been dubbed No-Signalling Quantum Random Access Code (NS-QRAC) \cite{grudka2015nonsignaling} (Row 5 of Table~\ref{Tab0}). The authors of \cite{grudka2015nonsignaling} considered a restricted classical channel and unbounded no-signalling resources and showed that it can be realised perfectly using PR boxes \cite{PR_box}. In this work, we establish a lower bound for the probability of success of NS-QRAC with unbounded entanglement resources (Row 6 of Table~(\ref{Tab0})). Then, we analyse two modifications to the NS-QRAC scenario. First, we consider a quantum channel instead of the classical one shared between the two parties (Row 7 of Table~(\ref{Tab0})) and, second, we consider a constrained entanglement scenario with unbounded classical communication, which we call Constrained-No-Signalling Quantum Random Access Code (CNS-QRAC). The latter (Row 8 of Table~(\ref{Tab0})) has not been considered in the literature before.

Therefore, in this work we provide lower quantum bounds for the probability of success of two different classes of QRACs. We show that the considered tasks are operationally equivalent to scenarios with constrained resources. 

\subsection{Summary of the results}

In Sec. \ref{sec:tele} we consider, as a warm-up, the task of quantum teleportation with constrained classical resources.
We show, using the notion of generalised Bell states, that the maximal fidelity of a teleported state equals $k/d^2$, where $d$ is the Hilbert space dimension and $k\leq d^2$ is the number of bits of classical communication transmitted by Alice to Bob. 

Sec. \ref{sec:NSQRAC} concerns the NS-QRAC scenario as presented in \cite{grudka2015nonsignaling}, in which Bob aims at reproducing at his output one of the two qubits possessed by Alice. In doing so, Bob is equipped with two bits of classical communication received from Alice, as well as two maximally entangled pairs. We show that this problem can be seen as a constrained teleportation task. We provide a quantum lower bound, $P_{\text{succ}}^{\text{QM}} \leq \frac{5}{8}$, for the success of such a task for the qubit case.

In Sec. \ref{sec:QRAC-SE}, we introduce and study 
the QRAC-SE setup introduced in \cite{tavakoli2021correlations}. It concerns the classical information to be encoded and decoded (like in the classical RAC) while using \emph{both} a quantum channel as well as a shared entanglement resource. The QRAC-SE brings together QRAC, which uses a quantum channel but no entanglement \cite{ambainis2008quantum}, and EA-RAC which involves entanglement but employs a classical channel \cite{pawlowski2010entanglement}. We show that this problem can be seen as a constrained dense-coding protocol, which is dual to the constrained quantum teleportation considered in the previous sections. Namely, here the parties have more classical input than they can send perfectly using qudit dense-coding \cite{liu2002general}. We provide and analyse the efficiency of such protocols which can be quantified in two ways by calculating 1) the  minimum probability of success of decoding either of two strings, each of which consists of two digits of base $d$ or, 2) the average probability of success of the protocol (over all possible strings). We show that in the qubit case both these measures coincide. The encoding by Alice utilises the roots of the generalised Pauli matrices, as well as Gray codes \cite{gray1953pulse} and its non-Boolean generalisations \cite{er1984generating,sharma1978m}, which is an example of a single distance code. We also present an analysis for higher dimensions, $d=3,4$ and show that the two measures of efficiency differ (in contrast to $d=2$) --- the interpretation of this fact is also discussed. Doriguello et al. \cite{doriguello2021quantum} have studied RAC variations extended to Boolean functions denoted by the prefix $f-$, including $f-$QRAC and $f-$EA-RAC. We show a proof of concept of extending QRAC-SE to encode Boolean functions of initial classical information called $f-$QRAC-SE.

In Sec. \ref{sec:NSQRACwithQRACSE} we revisit and provide a modification for the NS-QRAC scenario as presented in \cite{grudka2015nonsignaling}, in which Bob aims at reproducing at his output one of the two qubits possessed by Alice. In doing so, Bob is equipped with one qubit received from Alice, as well as three maximally entangled pairs. This modification is in some sense a truly quantum RAC problem since the information to be encoded as well as the channel shared by the parties is quantum. We show that the quantum lower bound for the success of such a task coincides with that of QRAC-SE studies in Sec. \ref{sec:QRAC-SE} ($P_{\text{succ}}^{\text{QM}} \approx 0.728$), which is a better bound than the one obtained in scenario in Sec. \ref{sec:NSQRAC}.

Finally, in Sec \ref{sec:constrainedententanglement}, we consider a second modification of the NS-QRAC scenario from \cite{grudka2015nonsignaling}.  
We study constrained no-signalling resources while allowing unbounded classical information to be sent from Alice to Bob --- we call these Constrained-No-Signalling Quantum Random Access Codes (CNS-QRAC). We provide an upper bound, $P_{\text{succ}}^{\text{QM}} \leq \frac{3}{4}$, for the success of such a task for the qubit case. We further generalise the protocol to the case of $N \geq 2$ inputs of $d$-level quantum systems and show that $P_{\text{succ}}^{\text{QM}} (d,N) \leq (N+d-1)/(dN)$. Furthermore, we discuss an `asymmetric' scenario in which the input quantum systems are chosen randomly with probabilities $\{p_i\}_{i=1}^N$. We present an algorithmic solution for this case using constraints coming from entanglement monogamy by exploiting the framework of universal asymmetric quantum cloning machines~\cite{kay2009optimal,WernerCloning}. However, as we describe later, these monogamy relations are understood in a `reverse order' --- with many senders and only one receiver.

We conclude, in Sec \ref{sec:outlook}, with a discussion and open questions related to the fidelity bounds for NS-QRACs which are simulable {\it via} quantum interactions, and possible future studies of the QRAC-SE.

\section{\label{sec:tele}A constrained teleportation scenario}
In this section, we introduce the necessary notation and, as a warm-up for our further considerations, 
we consider a constrained quantum teleportation task, where the parties share a classical channel with fewer inputs than required for perfect teleportation \cite{bennett_teleporting_1993}.

Following~\cite{HorodeckiMPRFidelity} we start by introducing two parameters describing channels and states --- the singlet fraction (or entanglement fidelity) $F$ of a state and the fidelity of a quantum channel $\Lambda$, called also the transmission fidelity $f$. We shall use these parameters for describing the efficiency of the constrained teleportation and later in the NS-QRAC game (Sec. \ref{sec:constrainedententanglement}). For an arbitrary bipartite quantum state $\rho$ the singlet fraction (entanglement fidelity ) \cite{schumacher1996sending} is its overlap with the maximally entangled state $|\psi^+\>=\frac{1}{\sqrt{d}}\sum_i |ii\>$:
\begin{equation}
\label{singlet_fraction}
F(\rho)=\<\psi^+|\rho|\psi^+\>.
\end{equation}
Considering the local action of a channel $\Lambda$ on half of the maximally entangled state the parties create a state $\rho_{\Lambda}$. The entanglement fidelity $F(\rho_{\Lambda})$ tells us how close the input and the output states are. When one fixes the channel $\Lambda$, the quantity $F(\rho_{\Lambda})$  is called the entanglement fidelity of the channel $\Lambda$ denoted as $F(\Lambda)$. 
On the other hand, we have also another quantity describing the quality of a channel transmission. For an arbitrary channel $\Lambda$ one can define the transmission fidelity $f$ given as
\begin{equation}
\label{trans_fidelity}
f(\Lambda)=\int \operatorname{d}\!\phi \, \<\phi|\Lambda \big(|\phi\>\<\phi| \big)|\phi\>,
\end{equation}
where the integral is taken over all pure states distributed uniformly with the Haar measure. Exploiting the property of invariance of entanglement and transmission fidelity with respect to averaging over $U\otimes U^*$, where $U$ is a unitary transformation and the star denotes complex conjugation, one obtains the following dependence for any channel $\Lambda$~\cite{HorodeckiMPRFidelity}:
\begin{equation}
f(\Lambda)=\frac{F(\Lambda)d+1}{d+1}\,,
\end{equation}
where $d$ is the size of the Hilbert space that $\Lambda$ acts on.

Now we are in a position to introduce the constrained teleportation procedure. In the usual scenario for quantum teleportation, as established in~\cite{bennett_teleporting_1993}, see Figure~\ref{fig:tele1}, Alice has access to all $d^2$ measurements, and the teleportation process is perfect --- the entanglement fidelity $F$ equals to 1.
\begin{figure}[h!]
 \centering
 \includegraphics[width=.8\linewidth]{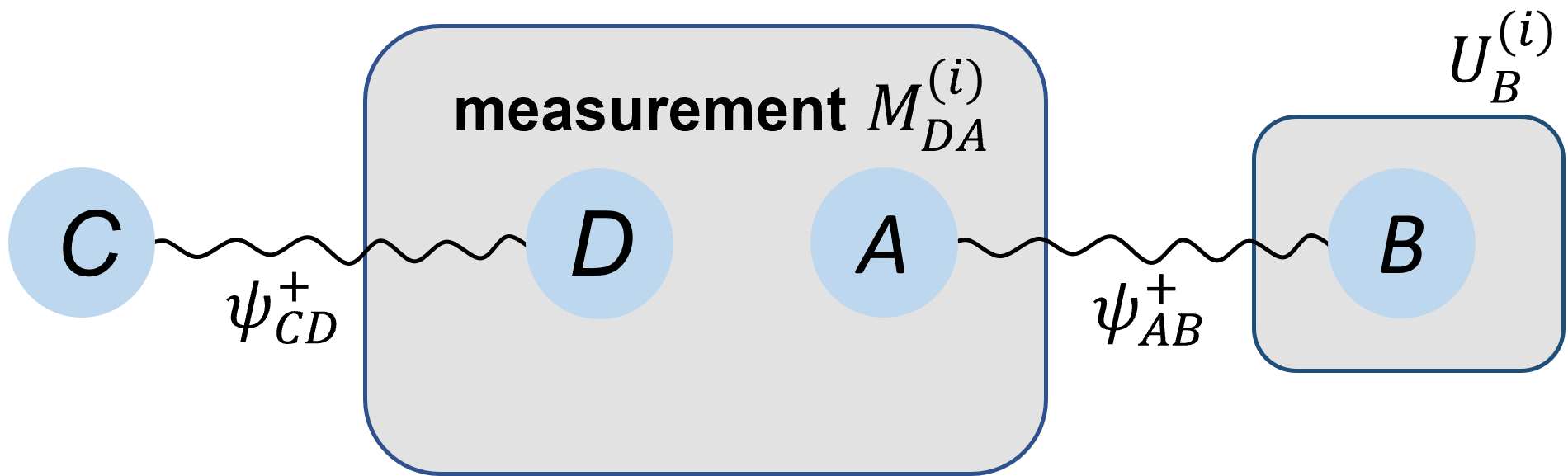}
 \caption{\label{fig:tele1}
 The schematic configuration for quantum teleportation task presented in the form of swapping of quantum correlations (this picture will be useful for further considerations). Two parties, Alice and Bob, share a $d-$dimensional maximally entangled state $\psi^{+}_{AB}$. 
 In addition, Alice has two particles in the maximally entangled state $\psi^{+}_{CD}$. She wishes to swap the correlations between 
 BC and DC in the sense that at the end the particles C and B should be maximally entangled like C and D were before the protocol.
 The measure of her success is the parameter $F$, called entanglement fidelity, calculated for the joint state of the particles C and B after the protocol (see the main text).
 To accomplish the task Alice applies a joint measurement (consisting of operators $M^{(i)}_{DA}$ corresponding to projections onto maximally entangled states) getting an outcome $a_i$, for $1\leq i\leq d^2$, transmitted later by a classical channel to Bob. To recover the state, Bob has to apply to system $B$ a unitary correction $U_B^{(i)}$, depending on the classical outcome $a_i$. The entanglement fidelity $F$ of the final state of the particles C and B, in this case, is maximal and equals 1  which means that after the protocol  Alice and Bob particles are indeed maximally entangled as intended.
 }
\end{figure}
In the constrained scenario, we assume that Alice's POVM (positive-operator valued measure) measurements $M_{DA}^{(i)}$ on systems $DA$ that have $k\leq d^2$ measurement operators, satisfying the standard relations:
\be
\label{cons}
\forall 1\leq i \leq k \qquad M_{DA}^{(i)}\geq 0,\qquad \sum_{i=1}^{k}M_{DA}^{(i)}=\mathbf{1}_{DA},
\ee
and we ask how well the parties can perform. 

After Alice's measurement $M_{DA}^{(i)}$ its outcome  is communicated to Bob by a classical channel.
Next, he chooses unitary operation $U_B^{(i)}$ depending on the transmitted outcome obtaining the following  unnormalized shared state:
\be
\label{state}
\varrho_{BC}^{(i)}= \tr_{AD} \left[U_B^{(i)}M_{DA}^{(i)}(\psi_{AB}^+\otimes \psi_{CD}^+)U_B^{(i) \dagger} \right], 
\ee
where $\psi_{AB}$ and $\psi_{CD}$ are maximally entangled states, and the final formula for the entanglement fidelity reads:
\begin{equation}
F( \{ M_{DA}^{(i)}, U_B^{(i)} \}) =\sum_{i=1}^{k} \tr(\varrho_{BC}^{(i)}\psi_{BC}^{+}).
\end{equation}
We would like to maximise each of the terms in the above sum, each of which is smaller or equal than one, 
i.e. our goal is to learn the following quantity:
\be
F_{\max} = \max_{ \{ M_{DA}^{(i)}, U_B^{(i)} \} } F ( \{ M_{DA}^{(i)}, U_B^{(i)} \}).
\ee
\begin{prop}
\label{coinstrainedF}
In the constrained teleportation protocol, when the sender has $k\leq d^2$ measurement outcomes, the maximal entanglement fidelity $F_{\max}$ of the protocol equals to
\begin{equation}
F_{\max} = \max_{ \{ M_{DA}^{(i)}, U_B^{(i)} \} } F ( \{ M_{DA}^{(i)}, U_B^{(i)} \})=\frac{k}{d^2}.
\end{equation}
\end{prop}
\begin{proof}
The proof is based on a straightforward calculation of the entanglement fidelity between the input and the output state. The sender has access to $k$ measurement operators in the POVM, acting on systems $DA$, satisfying equation~\eqref{cons}.
Then, denoting the teleportation channel from Alice to Bob by $\mathcal{N}$, the entanglement fidelity reads
\begin{equation}
\label{eq2}
\begin{split}
F&=\tr [\psi_{CB}^+(\mathbf{1}_C\otimes \mathcal{N}_D)\psi_{CD}^+]\\
&=\sum_{i=1}^{k}\tr [\psi^+_{CB}U_B^{(i)}M_{DA}^{(i)}(\psi^+_{CD}\ot \psi^+_{AB})(U_B^{(i)})^{\dagger}].
\end{split}
\end{equation}
Now, applying the so called `ping-pong' trick, which reads
\begin{equation}
\label{eq3}
(\mathbf{1}_A\ot X_B)\psi^+_{AB}=(X_A^t\ot \mathbf{1}_B)\psi^+_{AB},
\end{equation}
for an arbitrary operator $X$, with $t$ denoting a transposition, we rewrite the second line of~\eqref{eq2} as
\begin{equation}
\label{eq4}
\begin{split}
F&=\sum_{i=1}^{k}\tr [\psi^+_{CB}U_B^{(i)}(M_{CB}^{(i)})^{t}(\psi^+_{CD}\ot \psi^+_{AB})(U_B^{(i)})^{\dagger}]\\
&=\sum_{i=1}^{k}\frac{1}{d^2}\tr [\psi^+_{CB}U_B^{(i)}(M_{CB}^{(i)})^{t}(U_B^{(i)})^{\dagger}]\\
&=\sum_{i=1}^{k}\frac{1}{d^2}\tr [(U_B^{(i)})^{\dagger}\psi^+_{CB}U_B^{(i)}(M_{CB}^{(i)})^{t}].
\end{split}
\end{equation}
We are interested in the maximal value of the entanglement fidelity $F$ from~\eqref{eq4}, where the maximisation runs over all possible sets of measurements $\{M_{CB}^{(i)}\}_{i=1}^{k}$. To do so it is enough to choose for every $1\leq i\leq k-1$ unitary $U_B^{(i)}$ to be one of the generating unitary for the generalised Bell states. Indeed, to maximise~\eqref{eq4} we need to ensure that each term under the sum is equal to $1/d^2$. Notice that each $M_{CB}^{(i)}$ is bounded by the identity operator and $\psi^+_{CB}$ are rank-1 projectors. Using this observation we conclude that their overlap is bounded by $1/d^2$ already, and we must saturate the overlaps.   This can be done by defining the following mapping $(U_B^{(i)})^{\dagger}\psi^+_{CB}U_B^{(i)}\mapsto \psi^+_{CB}(i)$, where $\psi^+_{CB}(i)$ is the generalised Bell state. Now to maximise the entanglement fidelity, we choose the first $k-1$ measurements $(M_{CB}^{(i)})^{t}$ to exactly equal $\psi^+_{CB}(i)$, while the last one we take to be $(M_{CB}^{(k)})^{t}=\mathbf{1}-\sum_{j=1}^{k-1}\psi^+_{CB}(j)$. Plugging the above to~\eqref{eq4} we get
\begin{equation}
\label{eq5}
\begin{split}
F_{\max}&=\frac{1}{d^2}\sum_{i=1}^{k-1}\tr[\psi^+_{CB}(i)]+\\
&\qquad+\frac{1}{d^2}\tr \bigg[\psi^+_{CB}(k)\bigg(\mathbf{1}-\sum_{j=1}^{k-1}\psi^+_{CB}(j)\bigg)\bigg]\\
&=\frac{k-1}{d^2}+\frac{1}{d^2}=\frac{k}{d^2},
\end{split}
\end{equation}
since for all $1\leq j\leq k-1$ we have $\tr(\psi^+_{CB}(k)\psi^+_{CB}(j))=0$.
\end{proof}

\section{\label{sec:NSQRAC}Random access code with quantum inputs and output}

\subsection{No-Signalling Quantum Random Code Boxes (NS-QRAC)}

The NS-QRAC was introduced in \cite{grudka2015nonsignaling} as a quantum version of the random access code, with qubits instead of bits being encoded and randomly accessed. Consider two space-like separated parties, Alice and Bob, who share \textit{any} No-Signalling resource, which may be quantum or post-quantum, represented by the NS-QRAC Box. Alice has two qubits $\psi_1, \psi_2$ at her disposal and communicates to Bob two bits $a=(a_1,a_2)$ of classical information. Bob, using the received data aims at reproducing the qubit of his choice $\rho_{x} \vc \rho_{x,b=a}$, with $x \in \{1,2\}$ parametrising his decision\footnote{In \cite{grudka2015nonsignaling} Bob's decision was encoded in a quantum input $\omega_x$. However, we are interested in studying the efficiency of quantum information transmission from Alice to Bob, for which it is sufficient to consider a classical 1-bit input $x$.} --- see Fig. \ref{fig:RAC} 

\begin{figure}[h]
 \centering
 \includegraphics[width=1.0\linewidth]{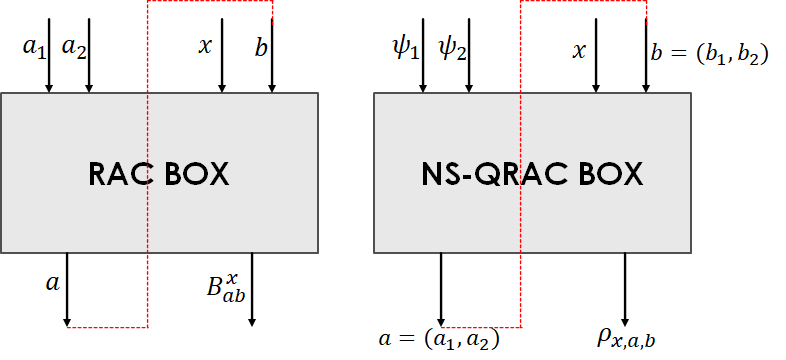}
 \caption{\label{fig:RAC}
 (Left panel): In the classical RAC Bob aims at reproducing at his output $B$ one of Alice's bits, $a_1$ or $a_2$, given his choice $x \in \{1,2\}$ and a single bit $a$ received from Alice. (Right panel) In the quantum analogue Bob seeks to reproduce at his output $\rho$ one of Alice's qubit, $\psi_1$ or $\psi_2$, of his choice $x \in \{1,2\}$ given two bits of information $a=(a_1,a_2)$ from Alice and a NS-QRAC box\cite{grudka2015nonsignaling}. 
 }
 \label{fig:NSQRAC}
\end{figure}

In \cite{grudka2015nonsignaling} it was shown that such a task can be perfectly realised when the NS-QRAC box consists of two maximally entangled states as well as two no-signalling post-quantum devices --- the Popescu--Rohrlich (PR) boxes \cite{PopescuRohrlich94}. This result relies on the fact that one can achieve perfect RAC using PR--boxes \cite{InformationCausality,popescu2014nonlocality}. On the other hand, it cannot be perfectly realised if Alice and Bob share only quantum no-signalling resources, though the quantum bound for the probability of success has not been quantified so far \cite{grudka2015nonsignaling}.

We ask what is the probability of success in the scenario where the NS-QRAC box consists only quantum no-signalling resources (that is, entanglement) and Bob is guessing one of two qubits (or, more generally, qudits), $\psi_1, \psi_2$. For fixed $\psi_1, \psi_2$ the probability of success is defined as
\begin{align}
\label{N2}
 P_{\text{succ}}^{\text{QM}} = \tfrac{1}{2} \big( F(\psi_1,\rho_1) + F(\psi_2,\rho_2) \big),
\end{align}
where $F$ is the quantum fidelity \cite{jozsa1994fidelity}.

\subsection{A quantum lower bound NS-QRAC}
 Consider the generalised scenario (as compared to \cite{grudka2015nonsignaling}), in which Alice has two qudits and Bob has a choice of which qudit to be teleported to him. Alice can send at most $2\log(d)$ bits to Bob. Alice and Bob share a NS-QRAC box consisting of two d-dimensional maximally entangled states $\psi^+_{\widetilde{A}_1 B_1}=\psi^+_{\widetilde{A}_2 B_2}=\frac{1}{\sqrt{d}}\sum_i^d |ii\rangle$. 

Note that if Alice could communicate $4\log(d)$ bits, rather than $2\log(d)$, then she could, using two shared maximally entangled states, teleport to Bob both qudits and hence Bob could perfectly recover the qudit of his choice.

Therefore, this scenario can be seen as a teleportation protocol with constrained communication channel, similar to the warm-up problem considered in Sec. \ref{sec:tele}. 
\begin{figure}[h!] \includegraphics[width=0.9\linewidth]{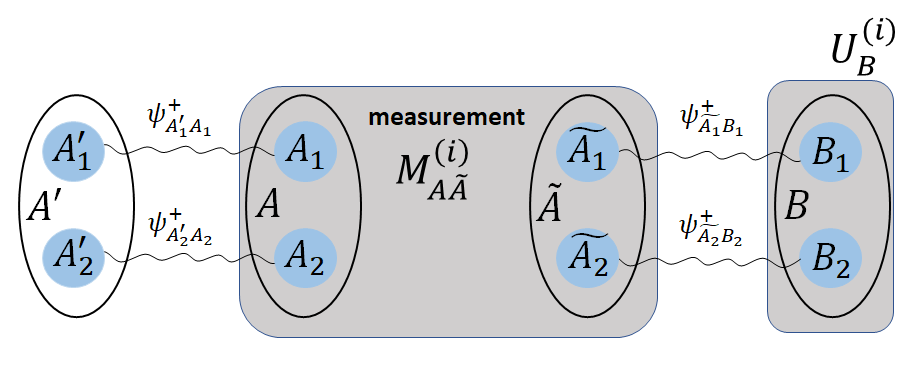}
 \caption{\label{twostatetele}The schematic description of the NS-QRAC scenario with two parties with quantum no-signalling resources, i.e. two $d$-dimensional maximally entangled states. In this situation Alice wishes to teleport either $\psi_1$ or $\psi_2$. Then Bob chooses which state of these two he wants to recover by applying respective unitaries $U_B^{(i,1)}, U_B^{(i,2)}$ depending on classical messages, indexed by $i$, sent by Alice.}
\end{figure}
We are interested in the fidelity and we use an analogous calculation to that in the Sec. \ref{sec:tele}

If Bob has to output \textit{both} qudits then one could straightforwardly use the result from Proposition~\ref{coinstrainedF} when considering the joint purified state $\psi^+_{A_1'A_1}\otimes \psi^+_{A_2'A_2} $, joint POVM measurement operators $M^{A_1\tilde{A_1}A_2\tilde{A_2}}_i$ for $i\in\{1,...,d^2\}$, joint shared entanglement $\psi^+_{\tilde{A_1}B_1}\otimes\psi^+_{\tilde{A_2}B_2}$ and joint corrective unitary $U_i^{B_1B_2}$. The joint entanglement fidelity of Bob having $\psi_1$ and $\psi_2$ from Proposition~\ref{coinstrainedF} would then be :
\begin{align}
\label{jointfidelity}
 F=\frac{d^2}{d^4}=\frac{1}{d^2}.
\end{align}
Note that the above result can be also achieved by considering product unitary correction $U_B$ of the form
\begin{align}
\label{unitaries}
 U_i^B=U^{B_1}_i \otimes U^{B_2}_i,\quad i=1,...,d^2,
\end{align}
and classically correlated measurements on Alice's side of the form
\begin{equation}\label{Mfactorise}
M_{A_1'A_2'B_1B_2}^{(i)}=M_{A_1'B_1}^{(i)}\otimes M_{A_2'B_2}^{(i)}.
\end{equation}

But Bob only needs to output one qudit of his choice for the NS-QRAC problem, and we are interested in evaluating equation~(\ref{N2}). Given that considering factorised measurements of the form in equation~(\ref{Mfactorise}) and factorised unitary correction of the form in equation~(\ref{unitaries}) was sufficient to find the joint fidelity equation~(\ref{jointfidelity}), we will consider these for our scenario and one can show that this is sufficient. Notice that this is equivalent to the constrained teleportation protocol from section~(\ref{sec:tele}) performed twice with a total of $d^2$($<d^4$) inputs since Alice can communicate $2\log(d)$ bits (or $d^2$ inputs labelled by $i$). If we assign $k'\leq d^2$ inputs to teleportation of $\psi_1$ and if we assign $(d^2-k')\leq d^2$ inputs to teleportation of $\psi_2$ upon application of equation~(\ref{eq5}) we have:

\begin{align}
&P_{\text{succ}}^{\text{QM}}\nonumber \\&\geq \frac{1}{2} \bigg( 
\sum_{i=1}^{k'}\operatorname{tr} \left[\psi^+_{A_1'B_1}U_{B_1}^{(i)}M_{A_1\widetilde{A}_1}^{(i)}\left(\psi^+_{A_1'A_1}\otimes \psi^+_{\widetilde{A}_1B_1}\right)U_{B_1}^{(i),\dagger}\right]\nonumber\\
&+
\sum_{i=k'+1}^{d^2}\operatorname{tr} \left[\psi^+_{A_2'B_2}U_{B_2}^{(i)}M_{A_2\widetilde{A_2}}^{(i)}\left(\psi^+_{A_2'A_2}\otimes \psi^+_{\widetilde{A}_2B_2}\right)U_{B_2}^{(i),\dagger}\right]
\bigg)\nonumber\\
& = \frac{1}{2} \left( \frac{k'}{d^2} + \frac{(d^2-k')}{d^2} \right) = \frac{1}{2}. \label{constrainedclassical}
\end{align}
Note that the value of $k'$ assigned to $\psi_1$ does not matter for the average success. 

One can do slightly better by employing another strategy with the same resources. In this case, Alice does the standard teleportation measurements on her end to receive $d^4$ inputs but sends the $d^2$ inputs associated with the first qudit. Then Bob's guessing probability for the first qudit would be $1$ while guessing the second qudit would be completely random with chance $\frac{1}{d^2}$. The average success probability would then be:
\begin{align}\label{eqn20}
\frac{1}{2}\left(1+\frac{1}{d^2}\right) \leq P_{\text{succ}}^{\text{QM}}.
\end{align}

These results are not surprising considering here the NS-QRAC box consists of two maximally entangled states between Alice and Bob while in \cite{grudka2015nonsignaling} it is shown that the NS-QRAC box which consists of two PR-boxes and two maximally entangled states can achieve $P_{\text{succ}}=1$. The presented protocol yields a lower bound for $P_{\text{succ}}^{\text{QM}}$ given a fixed amount of shared entanglement (two maximally entangled states). It is not clear whether this protocol is optimal and whether the bound \eqref{eqn20} can be improved with the increase of the amount of shared entanglement (although intuitively the latter seems unlikely given the difficulty of encoding $d^4$ inputs into $d^2$ inputs).

 We will see in Sec.~(\ref{sec:NSQRACwithQRACSE}) that if one modifies the NS-QRAC scenario so that Alice can send a quantum message (a qudit) instead of a classical message ($2\log(d)$ bits) one can have an improvement of the bound from equation~(\ref{eqn20}). To this end we will employ the QRAC-SEs, defined and studied in the following section
 
\section{\label{sec:QRAC-SE}Quantum Random Access Codes with Shared Entanglement}
In previous sections, we studied the NS-QRAC, which --- as we have shown --- can be seen as a constrained quantum teleportation task. In this section, we consider a complementary bipartite task -- a constrained quantum dense-coding task --
and we frame this task as a QRAC with shared entanglement (QRAC-SE). 

The (Q)RAC problem in the literature is denoted by $n \overset{p}{\mapsto} m$, where $n$ bits are encoded over $m$ (qu)bits, where any bit may be decoded with the probability of success \emph{at least} $p$, which should be more than the trivial guessing probability with no communication. The principle of Information Causality \cite{pawlowski2009information} states that the number of bits that can be perfectly decoded in an instance of QRAC is at most $m$. 
Since we also allow for shared entanglement one can send more bits perfectly (reminiscent to dense-coding). The QRAC problem of the form $n \overset{p}{\mapsto} 1$ has been studied in detail \cite{ambainis2008quantum} where $p$ is required to be greater than merely guessing, that is $p>0.5$. 

Given that we are interested in including shared entanglement into QRAC, the principle of Information Causality does not restrict us since we can utilise dense-coding, and thus we may decode multiple bits simultaneously. Therefore, we require a new general notation for QRAC-SE when $n$ digits of base $d$ are encoded over $m$ qudits with $d'$ dimension and $l$ shared-entanglement resources of $d'^2$ dimension, and where $k$ digits of base $d$ are decoded with probability at least $p$, presented as: 
\begin{align}
n_d \xmapsto{p, k_d} (m_{d'} , l_{d'}), 
\end{align} 
Now that the notation has been established, we present the problems we consider which are special cases of the above general form. Here digits of base $d$ are a generalisation to bits, that suit the task of dense-coding with qudits. 

\subsection{Problem Statement}

We consider the class of problems where we have two classical strings $\{a^{(0)},a^{(1)}\}$ each made of two digits of base $d$ which Alice sends over a qudit to Bob and both share the maximally entangled state $\ket{\psi_+}=\frac{1}{\sqrt{d}}\sum^{d-1}_{i=0}\ket{ii}$. Given the choice bit $c\in\{0,1\}$ unknown to Alice, Bob decodes $a^{(c)}$. We use the fact that two digits of base $d$ correspond to $d^2$ inputs in the notation.
This class of problems is denoted by:
\begin{align}\label{eqn:QRACSEd}
 2_{d^2} \xmapsto{p,1_{d^2}} (1_{d} , 1_{d}).
\end{align}
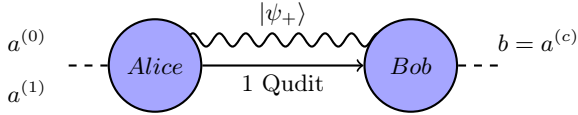
\begin{figure}[h!]
 \centering
 \begin{tikzpicture}[scale=1.7,auto,swap]
 \foreach \pos/\name in { {(0,-0.8)/Alice'}, {(2,-0.8)/Bob'},{(-0.7,-1.0)/a_i}}
 \node[vertex3] (\name) at \pos {};
 \foreach \source/ \dest /\weight in { Alice'/Bob'/}
 \path[edge5] (\source) -- node[weight] {$\weight$} (\dest);
 
 \foreach \pos/\name in { {(0,-1)/Alice}, {(2,-1)/Bob}}
 \node[vertex] (\name) at \pos {$\name$};

 \foreach \pos/\name in { {(1,-0.6)/psi }}
 \node[vertex3] (\name) at \pos {$|\psi_+\rangle$};
 
 \foreach \pos/\name in { {(-1.0,-0.8)/a0}}
 \node[vertex3] (\name) at \pos {$a^{(0)}$};

 \foreach \pos/\name in { {(-1.0,-1.2)/a1}}
 \node[vertex3] (\name) at \pos {$a^{(1)}$};
 
 \foreach \pos/\name in { {(3.0,-0.8)/b=a^{(c)}}}
 \node[vertex3] (\name) at \pos {$\name$};
 
 
 \foreach \pos/\name in { {(2.75,-1.0)/b}}
 \node[vertex3] (\name) at \pos {}; 
 
 \foreach \source/ \dest /\weight in { Alice/Bob/\text{1 Qudit}}
 \path[edge] (\source) -- node[weight] {$\weight$} (\dest);

 \foreach \source/ \dest /\weight in {a_i/Alice/, Bob/b/}
 \path[edge2] (\source) -- node[weight] {$\weight$} (\dest);
\end{tikzpicture}
 \caption{Schematic diagram for QRAC-SE problem of the class $2_{d^2} \xmapsto{p,1_{d^2}} (1_{d} , 1_{d})$, a bipartite task within which Alice encodes two strings $\{a^{(0)},a^{(1)}\}$ (each consisting of two digits of base $d$) over a qudit, which is part of the maximally entangled state $|\psi_+\rangle$ shared by Alice and Bob. Alice then sends the qudit and Bob performs a measurement on the entangled state to decode a string $a^{(c)}$ of his choice c.}
\label{fig:QRACSEd}
\end{figure}

For the purpose of this paper, the states encoded by Alice through applying local gates on the entangled state are pure $|\psi_e\rangle$ and the measurements by Bob correspond to projectors associated with the pure states $|\psi_d\rangle$. Each string value, as well as the choice $c$, is considered to be equally probable. The probability of success is then given by:
\begin{align}
 P \big(b=a^{(c)} \big| a^{(0)},a^{(1)},c \big) = |\langle\psi_d|\psi_e\rangle|^2.
\end{align}

There are two measures of success which we consider. We will see that they coincide for $d=2$ but will be different for higher dimensions.The first measure concerns the average success of the protocol over the choice $c$ and the strings $a^{(0)},a^{(1)}$,
\begin{align}
 P_{\text{avg}} &= \frac{1}{2d^4} \sum^{1}_{c=0}\sum^{d^2}_{a^{(0)}=0}\sum^{d^2}_{a^{(1)}=0} P(b=a^{(c)}|a^{(0)},a^{(1)},c).\label{Pavg} 
\end{align}

The second measure of success is defined similarly to the (Q)RAC problem, where any string may be decoded with the probability of success \emph{at least} $p$, more formally:
\begin{align}\label{Pmin}
 P_{\text{min}} = \min_{c \in \{0,1\}}\min_{a^{(c)} \in \{0,...,d^2\}} P(b=a^{(c)}|a^{(c)}),
\end{align}
where $P(b=a^{(c)}|a^{(c)})$ is given by 
\begin{align}
 P(b=a^{(c)}|a^{(c)})= \frac{1}{d^2}\sum^{d^2}_{a^{(\bar{c})}=0} P(b=a^{(c)}|a^{(0)},a^{(1)},c),
\end{align}
such that $a^{(\bar{c})}$ is the string not chosen to be decoded.

For the task of $2_{d^2} \xmapsto{p,1_{d^2} } (1_{d} , 1_{d})$, we want the probabilities to be greater than the \emph{trivial strategy} -- within which one would send one of the strings perfectly using qudit dense-coding and simply guess the other. The success for such a protocol would be $P_{\text{min}}=\frac{1}{d^2}$ and $P_{\text{avg}}=\frac{1}{2}\left( 1+ \frac{1}{d^2}\right)$. \\

\subsection{Qudit Dense-Coding}
Note that the class of problems: $2_{d} \xmapsto{p,1_{d}} (1_{d} , 1_{d})$, where two classical digits of base $d$ are encoded over a qudit part of a shared maximally entangled state $\ket{\psi_+}$ can be achieved perfectly due to the application of qudit dense-coding protocol introduced by Liu, Long, Tong and Li in \cite{liu2002general}. 

The qudit dense-coding protocol employs the generalised Pauli matrices which have the following action on the vectors of the computational basis $|k\rangle$:
\begin{equation}\label{XZ}
\begin{aligned}
&X|k\>=|k\oplus 1\>,\\
&Z|k\>=\operatorname{exp}\left(\frac{2\pi \operatorname{i}k}{d}\right)|k\>,
\end{aligned}
\end{equation}
such that $X^d=\mathbf{1}=Z^d$ and encodes two digits of base $d$ or $d^2$ inputs on the maximally entangled state as follows:
\begin{align}\label{actionU}
|\psi\>=\left(X^{a^{(0)}} Z^{a^{(1)}}\otimes \mathbf{1}\right)|\psi_+\>.
\end{align}
Since the encoded states span an orthonormal basis of size $d^2$, both digits can be sent and decoded perfectly. 

To tackle the difficult class of problems $2_{d^2} \xmapsto{p,1_{d^2}} (1_{d} , 1_{d})$ we will be utilising fractional powers of the generalised Pauli matrices
for Alice's encoding.

\subsection{The case $d=2$} We are now ready to tackle the $2_{d^2} \xmapsto{p,1_{d^2} } (1_{d} , 1_{d})$ for $d=2$, that is, the problem of encoding two 4-dimensional strings $a^{(0)},a^{(1)}$ (since we can write a 4-dimensional string as two bits we have $a^{(0)}\equiv \{a^{(0)}_0,a^{(0)}_1\}$ and $a^{(1)}\equiv \{a^{(1)}_0,a^{(1)}_1\}$ where $a^{(i)}_j$ are bits), sent over a qubit channel with a shared Bell state, where one of the strings is decided by Bob to be decoded $b=a^{(c)}$ (or $\{b_0,b_1\}=\{a^{(c)}_0,a^{(c)}_1\}$) given choice $c$. This is denoted by $2_4 \xmapsto{p,1_4} (1_{2} , 1_{2})$. We will show that the probability of success is $P_{\text{min}} = P_{\text{avg}} \approx 0.73$. One may also want to decode \emph{any} two of the four bits encoded over the same qubit channel and shared Bell state. This problem is denoted by $4_2 \xmapsto{p,2_2} (1_{2} , 1_{2})$. We now present an explicit protocol which gives a lower bound for the probabilities $P_{\text{min}} = P_{\text{avg}}$. This task ($2_4 \xmapsto{p,1_4} (1_{2} , 1_{2})$) has also been studied in \cite{piveteau2022entanglement} where they show the optimum success is bounded above by $0.75$.

\begin{figure}[h]
 \centering
 \begin{tikzpicture}[scale=1.7,auto,swap]
 \foreach \pos/\name in { {(0,-0.8)/Alice'}, {(2,-0.8)/Bob'},{(-0.7,-1.0)/a_i}}
 \node[vertex3] (\name) at \pos {};
 \foreach \source/ \dest /\weight in { Alice'/Bob'/}
 \path[edge5] (\source) -- node[weight] {$\weight$} (\dest);
 
 \foreach \pos/\name in { {(0,-1)/Alice}, {(2,-1)/Bob}}
 \node[vertex] (\name) at \pos {$\name$};

 \foreach \pos/\name in { {(1,-0.6)/psi }}
 \node[vertex3] (\name) at \pos {$|\psi_+\rangle$};
 
 \foreach \pos/\name in { {(-1.0,-0.8)/a0}}
 \node[vertex3] (\name) at \pos {$\{a^{(0)}_0,a^{(0)}_1\}$};

 \foreach \pos/\name in { {(-1.0,-1.2)/a1}}
 \node[vertex3] (\name) at \pos {$\{a^{(1)}_0,a^{(1)}_1\}$};
 
 \foreach \pos/\name in { {(3.0,-0.8)/b0b1}}
 \node[vertex3] (\name) at \pos {$\{b_0,b_1\}=$};
 
 \foreach \pos/\name in { {(3.0,-1.2)/ac0ac1}}
 \node[vertex3] (\name) at \pos {$\{a^{(c)}_0,a^{(c)}_1\}$};
 
 
 \foreach \pos/\name in { {(2.75,-1.0)/b}}
 \node[vertex3] (\name) at \pos {}; 
 
 \foreach \source/ \dest /\weight in { Alice/Bob/\text{1 Qubit}}
 \path[edge] (\source) -- node[weight] {$\weight$} (\dest);

 \foreach \source/ \dest /\weight in {a_i/Alice/, Bob/b/}
 \path[edge2] (\source) -- node[weight] {$\weight$} (\dest);
 
\end{tikzpicture}
 \caption{Schematic diagram for $d=2$ QRAC-SE problem $2_4 \xmapsto{p,1_4} (1_{2} , 1_{2})$, within which Alice encodes two strings of dimension $d^2=4$ (or two pairs of two bits $\{a^{(0)}_0,a^{(0)}_1\}$ and $\{a^{(1)}_0,a^{(1)}_1\}$) over a qubit part of a Bell pair $|\psi_+\rangle$, sends a qubit to Bob and Bob decodes a string $\{a^{(c)}_0,a^{(c)}_1\}$ of his choice c.}
\label{fig:QRACSE2}
\end{figure}
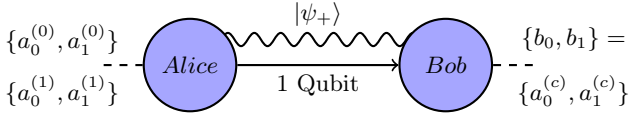
%

Our protocol relies on Alice's encoding using roots of the generalised Pauli matrices: $\sqrt{X}$ and $\sqrt{Z}$. We introduce the 4-dimension encoding strings $e^{(2)}_0$ and $e^{(2)}_1$ defined in Table~(\ref{Tab1}).
Notice that the encoding in $e^{(2)}_i$ follows Gray code ordering in $\{a^{(0)}_i,a^{(1)}_i\}$. One may recall  that Gray codes are also used in QRAC $2\xmapsto{p=0.85}1$ and this is desirable here since we wish to have a coding such that the codes pertaining to each encoded bit are close to each other in distance. While Gray codes are not always available for higher-dimensional bases nonetheless we will continue to use an encoding that employs maximum closeness for all encoded bits through a single distance code. It is important to note here that lexicographical codes always underperform the form of encoding we provide here. 
\begin{table}[h!]
 \begin{center}
 \begin{tabular}{|c|c|}
 \hline
 $e^{(2)}_i$ & $\{a^{(0)}_i,a^{(1)}_i\}$ \\
 \hline
 0 & \{0,0\} \\
 1 & \{0,1\} \\
 2 & \{1,1\} \\
 3 & \{1,0\} \\
 \hline
 \end{tabular}
 \caption{The encoding scheme by Alice involves a map between the two strings $a^{(0)}$ and $a^{(1)}$ and two encoding strings $e^{(2)}_0$ and $e^{(2)}_1$ where the superscript on $e$ denotes that we are considering $d=2$ case.}
 \label{Tab1}
 \end{center}
 \end{table}
 
The encoding by Alice in terms of $e^{(2)}_i$ is given by
 \begin{align}
 \ket{\psi_{e^{(2)}_0,e^{(2)}_1}}= (\sqrt{X})^{e^{(2)}_0}(\sqrt{Z})^{e^{(2)}_1} \otimes \mathbf{1} \ket{\psi_+},
 \end{align}
where $X,Z$ are Pauli matrices such that $X^2=Z^2=\mathbf{1}$. The encoding may also be presented visually 
as in Figure~(\ref{fig:aliceencodingd2}).

\begin{figure}[h!]
 \centering
 \begin{tikzpicture}[scale=1.7,auto,swap]

 \node[vertex3] at (-0.5,0.5) {\small $e^{(2)}_1\rightarrow$};
 \node[vertex3] at (-1,0.3) {\small $e^{(2)}_0\downarrow$};

 \draw (0-0.5,0.25) -- (3+0.5,0.25);
 \draw (-0.5,0+0.25) -- (-0.5,-1.5-0.25);
 \draw (0-0.5,-1.75) -- (3+0.5,-1.75);
 \draw (3.5,0+0.25) -- (3.5,-1.5-0.25);

 \foreach \pos/\name in { {(0,0.5)/00}, {(1,0.5)/01}, {(2,0.5)/11}, {(3,0.5)/10}}
 \node[vertex3] (\name) at \pos {\small $\name$};
 
 \foreach \pos/\name in { {(-1.,0)/00}, {(-1.,-0.5)/01}, {(-1.,-1)/11}, {(-1.,-1.5)/10}}
 \node[vertex3] (\name) at \pos {\small $\name$}; 
 
 \foreach \pos/\name in { {(0,0)/X^{0.0}Z^{0.0}}, {(1,0)/X^{0.0}Z^{0.5}}, {(2,0)/X^{0.0}Z^{1.0}}, {(3,0)/X^{0.0}Z^{1.5}}}
 \node[vertex3] (\name) at \pos {\small $\name$};

 \foreach \pos/\name in { {(0,-0.5)/X^{0.5}Z^{0.0}}, {(1,-0.5)/X^{0.5}Z^{0.5}}, {(2,-0.5)/X^{0.5}Z^{1.0}}, {(3,-0.5)/X^{0.5}Z^{1.5}}}
 \node[vertex3] (\name) at \pos {\small $\name$};
 
 \foreach \pos/\name in { {(0,-1)/X^{1.0}Z^{0.0}}, {(1,-1)/X^{1.0}Z^{0.5}}, {(2,-1)/X^{1.0}Z^{1.0}}, {(3,-1)/X^{1.0}Z^{1.5}}}
 \node[vertex3] (\name) at \pos {\small $\name$};
 
 \foreach \pos/\name in { {(0,-1.5)/X^{1.5}Z^{0.0}}, {(1,-1.5)/X^{1.5}Z^{0.5}}, {(2,-1.5)/X^{1.5}Z^{1.0}}, {(3,-1.5)/X^{1.5}Z^{1.5}}}
 \node[vertex3] (\name) at \pos {\small $\name$};
\end{tikzpicture}
 \caption{Visual representation of Alice's encoding in terms of the generalised Pauli matrices aiding Table~(\ref{Tab1}) for $d=2$}
 \label{fig:aliceencodingd2}
\end{figure}

Bob measures using projectors spanned by the basis:
\begin{align} 
 \ket{\psi_{b_0,b_1}} = \ (X^{(-1)^c b_0+\frac{1-2c}{4}} Z^{(-1)^c b_1+\frac{1-2c}{4}} \otimes \mathbf{1}) \ket{\psi_+} \label{bobmeasured2}
\end{align} where $b_0, b_1 \in \{0,1\}$ are Bob's guesses associated with the relevant projectors. We visually represent these measurements for the case $c=0$ and $c=1$ in figure~(\ref{fig:bobencodingd2}). Note that Bob's measurement basis for $c=0$ and $c=1$ have minimum possible overlap.

\begin{figure}[h]
 \centering
 \begin{tikzpicture}[scale=1.7,auto,swap]
 
 \node[vertex3] at (-0.5,0.5) {\small $e^{(2)}_1\rightarrow$};
 \node[vertex3] at (-1,0.3) {\small $e^{(2)}_0\downarrow$};
 
 \foreach \pos/\name in { {(0.5,-0.25)/b00},{(2.5,-0.25)/b01},{(0.5,-1.25)/b10},{(2.5,-1.25)/b11}}
 \node[fillc0] (\name) at \pos {};

 \foreach \pos/\name in { {(0,0.5)/00}, {(1,0.5)/01}, {(2,0.5)/11}, {(3,0.5)/10}}
 \node[vertex3] (\name) at \pos {\small $\name$};
 
 \foreach \pos/\name in { {(-1.,0)/00}, {(-1.,-0.5)/01}, {(-1.,-1)/11}, {(-1.,-1.5)/10}}
 \node[vertex3] (\name) at \pos {\small $\name$}; 
 
 \foreach \pos/\name in { {(0,0)/X^{0.0}Z^{0.0}}, {(1,0)/X^{0.0}Z^{0.5}}, {(2,0)/X^{0.0}Z^{1.0}}, {(3,0)/X^{0.0}Z^{1.5}}}
 \node[vertex3] (\name) at \pos {\small $\name$};

 \foreach \pos/\name in { {(0,-0.5)/X^{0.5}Z^{0.0}}, {(1,-0.5)/X^{0.5}Z^{0.5}}, {(2,-0.5)/X^{0.5}Z^{1.0}}, {(3,-0.5)/X^{0.5}Z^{1.5}}}
 \node[vertex3] (\name) at \pos {\small $\name$};
 
 \foreach \pos/\name in { {(0,-1)/X^{1.0}Z^{0.0}}, {(1,-1)/X^{1.0}Z^{0.5}}, {(2,-1)/X^{1.0}Z^{1.0}}, {(3,-1)/X^{1.0}Z^{1.5}}}
 \node[vertex3] (\name) at \pos {\small $\name$};
 
 \foreach \pos/\name in { {(0,-1.5)/X^{1.5}Z^{0.0}}, {(1,-1.5)/X^{1.5}Z^{0.5}}, {(2,-1.5)/X^{1.5}Z^{1.0}}, {(3,-1.5)/X^{1.5}Z^{1.5}}}
 \node[vertex3] (\name) at \pos {\small $\name$}; 

 \node at (0.5,-0.25)[circle,fill,inner sep=1pt]{};
 \node at (2.5,-0.25)[circle,fill,inner sep=1pt]{};
 \node at (0.5,-1.25)[circle,fill,inner sep=1pt]{};
 \node at (2.5,-1.25)[circle,fill,inner sep=1pt]{};
 
\end{tikzpicture}
 \begin{tikzpicture}[scale=1.7,auto,swap]
 
 \node[vertex3] at (-0.5,0.5) {\small $e^{(2)}_1\rightarrow$};
 \node[vertex3] at (-1,0.3) {\small $e^{(2)}_0\downarrow$};
 
 \foreach \pos/\name in { {(0,0.5)/0,0}, {(1,0.5)/0,1}, {(2,0.5)/1,1}, {(3,0.5)/1,0}}
 \node[vertex3] (\name) at \pos {\small $\name$};
 
 \foreach \pos/\name in { {(-1.,0)/0,0}, {(-1.,-0.5)/0,1}, {(-1.,-1)/1,1}, {(-1.,-1.5)/1,0}}
 \node[vertex3] (\name) at \pos {\small $\name$}; 

 \node at (-0.39,0.19)[circle,fill,inner sep=1pt]{};
 \node at (1.5,0.19)[circle,fill,inner sep=1pt]{};
 \node at (-0.39,-0.75)[circle,fill,inner sep=1pt]{};
 \node at (1.5,-0.75)[circle,fill,inner sep=1pt]{}; 
 
 \clip (-0.4,-1.7) rectangle (3.4,0.2);
 
 \foreach \pos/\name in { {(1.5,-0.75)/b11},
 {(-0.5,-0.75)/b10_1},{(3.5,-0.75)/b10_2},
 {(1.5,0.25)/b01_1},{(1.5,-1.75)/b01_2},
 {(-0.5,-1.75)/b00_1},{(3.5,-1.75)/b00_2},{(-0.5,0.25)/b00_3},{(3.5,0.25)/b00_4}}
 \node[fillc1] (\name) at \pos {};
 
 \foreach \pos/\name in { {(0,0)/X^{0.0}Z^{0.0}}, {(1,0)/X^{0.0}Z^{0.5}}, {(2,0)/X^{0.0}Z^{1.0}}, {(3,0)/X^{0.0}Z^{1.5}}}
 \node[vertex3] (\name) at \pos {\small $\name$};

 \foreach \pos/\name in { {(0,-0.5)/X^{0.5}Z^{0.0}}, {(1,-0.5)/X^{0.5}Z^{0.5}}, {(2,-0.5)/X^{0.5}Z^{1.0}}, {(3,-0.5)/X^{0.5}Z^{1.5}}}
 \node[vertex3] (\name) at \pos {\small $\name$};
 
 \foreach \pos/\name in { {(0,-1)/X^{1.0}Z^{0.0}}, {(1,-1)/X^{1.0}Z^{0.5}}, {(2,-1)/X^{1.0}Z^{1.0}}, {(3,-1)/X^{1.0}Z^{1.5}}}
 \node[vertex3] (\name) at \pos {\small $\name$};
 
 \foreach \pos/\name in { {(0,-1.5)/X^{1.5}Z^{0.0}}, {(1,-1.5)/X^{1.5}Z^{0.5}}, {(2,-1.5)/X^{1.5}Z^{1.0}}, {(3,-1.5)/X^{1.5}Z^{1.5}}}
 \node[vertex3] (\name) at \pos {\small $\name$};
 
 \node at (-0.39,0.19)[circle,fill,inner sep=1pt]{};
 \node at (1.5,0.19)[circle,fill,inner sep=1pt]{};
 \node at (-0.39,-0.75)[circle,fill,inner sep=1pt]{};
 \node at (1.5,-0.75)[circle,fill,inner sep=1pt]{}; 
 
\end{tikzpicture}
 \caption{Visual representation of Bob's decoding for choice $c=0$ (above) and $c=1$ (below). The black dots correspond to the projectors and thus four measurement outcomes and the rectangular shaded region displays which encoded states are associated with with measurement outcomes given by Eq. \eqref{bobmeasured2}}
\label{fig:bobencodingd2}
\end{figure}

Due to the symmetry between the cases $c=0$ and $c=1$, as well as between the bits, each term is equal to 
\begin{align}
 P(a^{(c)}_0,a^{(c)}_1)=|\bra{\psi_+}X^{\pm0.25}Z^{\pm0.25}\ket{\psi_+}|^2, \forall c,a^{(c)}_i
\end{align}
 and thus $P_{\text{avg}}=P_{\text{min}}= \left(\frac{1}{2} \left( 1 + \frac{1}{\sqrt{2}}\right)\right)^2 \approx 0.73$.
 
\noindent
Note that at Bob's end one may \emph{also} recover $\{a^{(0)}_0,a^{(1)}_1\}$, with measurement $(X^{b_0+\frac{1}{4}} Z^{b_1-\frac{1}{4}} \otimes \mathbf{1}) \ket{\psi_+}$ or $\{a^{(0)}_1,a^{(1)}_0\}$, with measurement $(X^{b_0-\frac{1}{4}} Z^{b_1+\frac{1}{4}} \otimes \mathbf{1}) \ket{\psi_+}$. We utilise this by considering the related problem of $4_2 \xmapsto{p,2_2} (1_{2} , 1_{2})$, which is harder than $2_4 \xmapsto{p,1_4} (1_{2} , 1_{2})$. Here we have four bits $\{a_0,a_1,a_2,a_3\}$ (previously we had $\{a^{(0)}_0,a^{(0)}_1,a^{(1)}_0,a^{(1)}_1\}$) and we wish to guess any two of these four bits. If we use the same encoding we used for $2_4 \xmapsto{p,1_4} (1_{2} , 1_{2})$ we have $P_{a_0,a_1}=P_{a_2,a_3}=P_{a_0,a_3}=P_{a_1,a_2}\approx0.728$. Since ${a_0,a_2},{a_1,a_3}$ are encoded as rows and columns, we choose to measure one of the bits as determined by the above measurements and guess the other to get $P_{a_0,a_2}=P_{a_1,a_3}\approx0.364$. Thus we have $P_{\text{avg}}\approx0.607$ and $P_{\text{min}}\approx0.364$. For the trivial strategy let us say that one sends $a_0,a_1$ perfectly and guesses the rest, then we have $P_{a_0,a_1}=1, P_{a_0,a_2}=P_{a_0,a_3}=P_{a_1,a_2}=P_{a_1,a_3}=0.5,P_{a_2,a_3}=0.25$ and thus $P_{\text{avg}}\approx0.542$ and $P_{\text{min}}=0.25$. 

\subsection{QRAC-SE for Boolean functions}

The problem of $f-$RAC and its variations ($f-$QRAC, $f-$EARAC) were studied in \cite{doriguello2021quantum}, where instead of guessing the $n$ bits Alice has, the task is to encode Boolean functions $f:\{0,1\}^k \rightarrow \{0,1\}$ defined over $k$ of $n$ bits that Alice has. We briefly discuss a proof of concept for generalising $f-$QRAC to study $f-$QRAC-SE. Consider the $f-$QRAC $4_2 \xmapsto{p,1_2}(1_{2} , 1_{2})$, where $f:\{0,1\}^3 \rightarrow \{0,1\}$, i.e. $k=3$, gives us four new bits to encode. We can use the same encoding and measurements as given for $4_2 \xmapsto{p,2_2}(1_{2} , 1_{2})$ to map the $f-$QRAC-SE problem of $4_2 \xmapsto{p,1_2}(1_{2} , 1_{2})$ to the QRAC-SE problem of $4_2 \xmapsto{p,1_2} (1_{2} , 1_{2})$ (which is simpler than $4_2 \xmapsto{p,2_2} (1_{2} , 1_{2})$ as well as $2_4 \xmapsto{p,1_4} (1_{2} , 1_{2})$. Then we can guess any bit with $P_{\text{avg}}=P_{\text{min}}\approx0.728$.

\subsection{Higher dimensions}

We have considered some variations of the $d=2$ case for the class of problems $2_{d^2} \xmapsto{p,1_{d^2} } (1_{d} , 1_{d})$, and now discuss a generalisation for higher $d$. Upon exploring higher dimensions solutions we observe that the encoding becomes combinatorially difficult. The solutions do not retain the same symmetry between $c=0,1$ leading to $P_{\text{avg}}\ne P_{\text{min}}$. Nonetheless, we present a method that generalises the protocols described so far that would help to find lower bounds (should they beat the trivial strategy success probabilities) and then provide explicit protocols\footnote{The probabilities are computed using Mathematica codes that can be found at \url{https://github.com/nitica/QRAC-SE}} for $d=3$ and $d=4$. 

The protocol would involve Alice's encoding using $d^{th}$ roots of the generalised Pauli matrices: $\sqrt[d]{X}$ and $\sqrt[d]{Z}$. The goal is to encode two strings $a^{(0)},a^{(1)}$, each of size $d^2$, which can be expressed instead as two strings consisting of 2 digits of base $d$: $a^{(i)}\equiv\{a^{(i)}_0,a^{(i)}_1\}$ where $a^{(i)}_j \in \{0,...,d-1\}$. The incomplete step involves defining the encoding through the map between $d^2$-dimensional strings $e^{(d)}_0,e^{(d)}_1$ and $a^{(0)},a^{(1)}$. We define this map partially using a notion of single distance code in Table~(\ref{Tab2}), which resemble the Gray codes used for $d=2$. The problem is to complete this encoding chart and provide explicit protocols and bounds. The encoding by Alice in terms of $e^{(d)}_i$ then is given by
 \begin{align}
 \ket{\psi_{e^{(d)}_0,e^{(d)}_1}}= \left( (\sqrt[d]{X})^{e^{(d)}_0}(\sqrt[d]{Z})^{e^{(d)}_1} \otimes \mathbf{1} \right) \ket{\psi_+},
 \end{align}
where $X,Z$ are Pauli matrices such that $X^d=Z^d=\mathbf{1}$.

\begin{table}
 \begin{center}
 \begin{tabular}{|c|c|}
 \hline
 $e^{(d)}_i$ & $\{a^{(0)}_i,a^{(1)}_i\}$ \\
 \hline
 $0$ & $\{0,0\}$ \\
 $\vdots$ & $\vdots$ \\
 $d-1$ & $\{0,d-1\}$ \\
 $d$ & $\{1,d\}$ \\
 $\vdots$ & $\vdots$ \\
 $2d-1$ & $\{1,x\}$ \\
 $2d$ & $\{2,x\}$ \\
 $\vdots$ & $\vdots$ \\
 $d^2-d-1$ & $\{d-2,y\}$ \\
 $d^2-d$ & $\{d-1,y\}$ \\
 $\vdots$ & $\vdots$ \\
 $d^2-1$ & $\{d-1,0\}$ \\
 \hline
 \end{tabular}
 \caption{Table presents $d^2$-dimensional strings $e^{(d)}_0$ and $e^{(d)}_1$. The proposed encoding scheme to be employed by Alice for the general problem from equation~(\ref{eqn:QRACSEd}) involves a map between the two strings $a^{(0)}$ and $a^{(1)}$ and two encoding strings $e^{(d)}_0$ and $e^{(d)}_1$ where the superscript on $e$ denotes that we are considering the general $d$ dimensional case. We subsequently provide explicit encoding for $d=3,4$}
 \label{Tab2}
 \end{center}
 \end{table}

Bob measures using projectors spanned by the basis:
\begin{align}
 \ket{\psi_{b_0,b_1}}\equiv \ (X^{(-1)^c b_0+\frac{1-c}{2}-\frac{1}{2d}} Z^{(-1)^c b_1+\frac{1-c}{2}-\frac{1}{2d}} \otimes \mathbf{1}) \ket{\psi_+},
\end{align}
where $b_0$ and $b_1 \in \{0,\ldots,d-1\}$ are Bob's guesses associated with the relevant projectors. Apart from the encoding by Alice the protocol is completely described. We provide explicit encoding charts by Alice for $d=3$ and $d=4$ below.

\subsubsection{The case $d=3$.}

The protocol relies on Alice's encoding using cube root of the generalised Pauli matrices: $\sqrt[3]{X}$ and $\sqrt[3]{Z}$. We introduce the 9-dimensional strings $e^{(3)}_0$ and $e^{(3)}_1$ defined in Table~(\ref{Tab3}) (left).

\begin{table}
 \begin{center}
 \begin{tabular}{|c|c|}
 \hline
 $e^{(3)}_i$ & $\{a^{(0)}_i,a^{(1)}_i\}$ \\
 \hline
 0 & \{0,0\} \\
 1 & \{0,1\} \\
 2 & \{0,2\} \\
 3 & \{1,2\} \\
 4 & \{1,0\} \\
 5 & \{1,1\} \\
 6 & \{2,1\} \\
 7 & \{2,2\} \\
 8 & \{2,0\} \\
 \hline
 \end{tabular}
 \quad
 \begin{tabular}{|c|c|}
 \hline
 $e^{(4)}_i$ & $\{a^{(0)}_i,a^{(1)}_i\}$ \\
 \hline
 0 & \{0,0\} \\
 1 & \{0,1\} \\
 2 & \{0,2\} \\
 3 & \{0,3\} \\
 4 & \{1,3\} \\
 5 & \{1,0\} \\
 6 & \{1,1\} \\
 7 & \{1,2\} \\
 8 & \{2,2\} \\
 9 & \{2,3\} \\
 10 & \{2,0\} \\
 11 & \{2,1\} \\
 12 & \{3,1\} \\
 13 & \{3,2\} \\
 14 & \{3,3\} \\
 15 & \{3,0\} \\
 \hline
 \end{tabular}
 \caption{Table presents the 9-dimensional strings $e^{(3)}_0$ and $e^{(3)}_1$ (left) and the 16-dimensional strings $e^{(4)}_0$ and $e^{(4)}_1$ (right). The encoding scheme by Alice involves a map between the two strings $a^{(0)}$ and $a^{(1)}$ and two encoding strings $e^{(d)}_0$ and $e^{(d)}_1$ where the superscript on $e$ denotes the dimension $d=3$ (left) and $d=4$ (right)}
 \label{Tab3}
 \end{center}
 \end{table}

The encoding by Alice in terms of $e^{(3)}_i$ is given by
 \begin{align}
 \ket{\psi_{e^{(3)}_0,e^{(3)}_1}}= \left( (\sqrt[3]{X})^{e^{(3)}_0}(\sqrt[3]{Z})^{e^{(3)}_1} \otimes \mathbf{1} \right) \ket{\psi_+},
 \end{align}
where $X,Z$ are the generalised Pauli matrices given by Eq. \eqref{XZ}, which satisfy $X^3=Z^3=\mathbf{1}$.

For $c=0$ Bob measures in the orthogonal basis $\{b_0,b_1\}\equiv \ (X^{b_0+\frac{1}{3}} Z^{b_1+\frac{1}{3}} \otimes \mathbf{1}) \ket{\psi_+}$, where $b_0$ and $b_1 \in \{0,1,2\}$, with the probability of success $P(b=a^{(0)})\approx0.582$. For $c=1$ Bob measures in the orthogonal basis $\{b_0,b_1\}\equiv \ (X^{-b_0-\frac{1}{6}} Z^{-b_1-\frac{1}{6}} \otimes \mathbf{1}) \ket{\psi_+}$, where $b_0$ and $b_1 \in \{0,1,2\}$, with the probability of success $P(b=a^{(1)})\approx0.386$.

\subsubsection{The case $d=4$.}
Similarly as in the 3-dimensional case, the protocol now bases on Alice's encoding via fourth root of the generalised Pauli matrices: $\sqrt[4]{X}$ and $\sqrt[4]{Z}$. We introduce the 16-dimensional strings $e^{(4)}_0$ and $e^{(4)}_1$ defined as in Table~(\ref{Tab3}) (right):

The encoding by Alice in terms of $e^{(4)}_i$ is given by
 \begin{align}
 \ket{\psi_{e^{(4)}_0,e^{(4)}_1}}= ((\sqrt[4]{X})^{e^{(4)}_0}(\sqrt[4]{Z})^{e^{(4)}_1}\otimes\mathbf{1})\ket{\psi_+},
 \end{align}
where $X,Z$ are given by Eq. \eqref{XZ} and $X^4=Z^4=\mathbf{1}$.

For $c=0$ Bob measures in the orthogonal basis $\{b_0,b_1\}\equiv \ (X^{b_0+\frac{3}{8}} Z^{b_1+\frac{3}{8}} \otimes \mathbf{1}) \ket{\psi_+}$, where $b_0$ and $b_1 \in \{0,1,2,3\}$, with the probability of success $P(b=a^{(0)})\approx0.629$. For $c=1$ Bob measures in the orthogonal basis $\{b_0,b_1\}\equiv \ (X^{-b_0-\frac{1}{8}} Z^{-b_1-\frac{1}{8}} \otimes \mathbf{1}) \ket{\psi_+}$, where $b_0$ and $b_1 \in \{0,1,2,3\}$, with the probability of success $P(b=a^{(1)})\approx0.261$.

\subsubsection{QRAC-SE summary}

To summarise, for the class of problems $2_{d^2} \xmapsto{p,1_{d^2} } (1_{d} , 1_{d})$  we found lower bounds for two different measures of success probability, through explicit protocols. The results are summarised in Table~(\ref{Tab4}). 
\begin{table}[h!]
\begin{center}
 \begin{tabular}{|c||c|c|c|c|}
 \hline
 $d$ & $P_{\text{min}}$ & $P^{\text{trivial}}_{\text{min}} = \frac{1}{d^2}$ & $P_{\text{avg}}$ & $P^{\text{trivial}}_{\text{avg}}=\frac{1}{2} (1+\frac{1}{d^2})$ \\
 \hline
 2 & $0.728$ & $0.250$ & $0.728$ & $0.625$ \\
 3 & $0.424$ & $0.111$ & $0.539$ & $0.556$ \\
 4 & $0.261$ & $0.063$ & $0.445$ & $0.531$ \\
 \hline
 \end{tabular}
 \caption{The lower bounds for probabilities (given by Eqs. \eqref{Pmin} and \eqref{Pavg}) for the class of problems QRAC-SE $2_{d^2} \xmapsto{p,1_{d^2} } (1_{d} , 1_{d})$, 
 for $d=2,3,4$.}
 \label{Tab4}
 \end{center} 
 \end{table} 
 
Note that for $d=2$, the probabilities $P_{\text{avg}}=P_{\text{min}}$ coincide, but for higher $d$ the trivial strategy performs better for $P_{\text{avg}}$, while the provided protocols perform better for $P_{\text{min}}$. While $P_{\text{avg}}$ measures how well does a protocol perform overall, $P_{\text{min}}$ is more focused at minimising the error. This explains the need for multiple measures of success probabilities and in different situations different success measures may be relevant. 
 
We have also established some further lower bounds, presented in Table~(\ref{T6}), for other variants with $d=2$.
 \begin{table}[h!]
 \begin{center}
 \begin{tabular}{|c||c|c|c|c|}
 \hline
  & $P_{\text{min}}$ & $P^{\text{trivial}}_{\text{min}} $ & $P_{\text{avg}}$ & $P^{\text{trivial}}_{\text{avg}}$ \\
 \hline
 $2_4 \xmapsto{p,1_4} (1_{2} , 1_{2})$ & $0.728$ & $0.250$ & $0.728$ & $0.625$ \\
 $4_2 \xmapsto{p,2_2} (1_{2} , 1_{2})$ & $0.364$ & $0.250$ & $0.604$ & $0.542$ \\
 $4_2 \xmapsto{p,1_2} (1_{2} , 1_{2})$ & $0.728$ & $0.5$ & $0.728$ & $0.75$ \\
 $4_2 \xmapsto{p,1_2 (f)} (1_{2} , 1_{2})$ & ,, & ,, & ,, & ,, \\
 \hline
 \end{tabular}
 \caption{The lower bounds for the success probabilities quantifying some variations of QRAC -SE for $d=2$, including the proof of concept for $f-$QRAC-SE (referred to as $4_2 \xmapsto{p,1_2 (f)} (1_{2} , 1_{2})$).}
 \label{T6}
 \end{center} 
 \end{table}
 
\section{\label{sec:NSQRACwithQRACSE}Quantum Bound for NS-QRAC using QRAC-SE}

Let us present a variation of the NS-QRAC discussed in Sec. \ref{sec:NSQRAC}. As previously, Alice and Bob share some no-signalling resource and Alice has two qudits $\psi_1, \psi_2$ at her disposal, but now Alice communicates to Bob one qudit $\phi$ instead of $2\log(d)$ bits. Bob, using the received data aims at reproducing the qudit of his choice $\rho_{x}$, with $x \in \{1,2\}$ parametrising his decision --- see Figure~(\ref{fig:modifiedNSQRAC}).

One may regard this game as a `truly quantum' NS-QRAC scenario in the sense that we are now trying to encode two qudits in one qudit --- exactly as in the classical RAC, where we encode many bits in one bit of information. 

\begin{figure}[h!]
 \centering
 \includegraphics[width=0.5\linewidth]{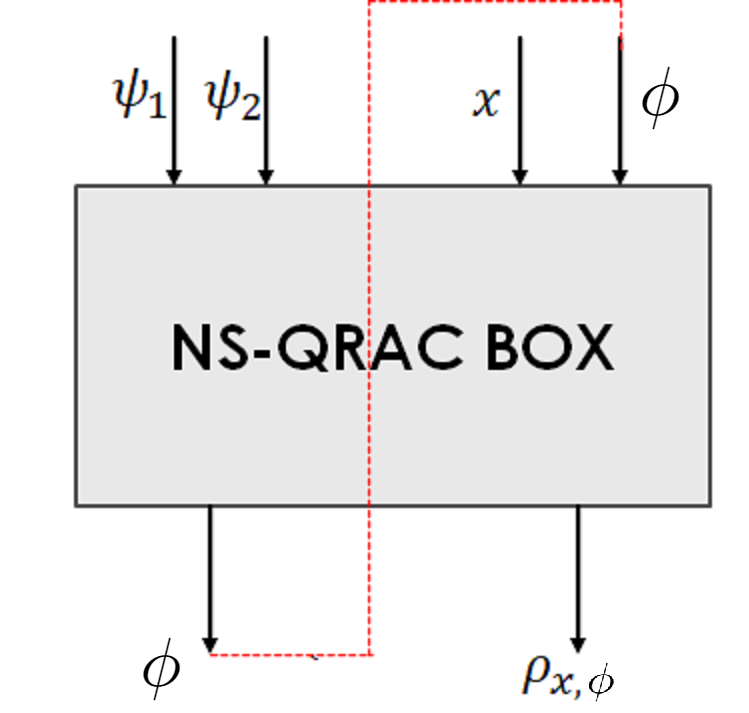}
 \caption{
 In the quantum analogue Bob seeks to reproduce at his output $\rho$ one of Alice's qubit, $\psi_1$ or $\psi_2$, of his choice $x \in \{1,2\}$ given one qubit of information $\phi$ from Alice and a NS-QRAC box \cite{grudka2015nonsignaling}.  }
 \label{fig:modifiedNSQRAC}
\end{figure}

We consider the strategy when Alice and Bob share three maximally entangled states through the NS-QRAC box. Alice performs the usual teleportation measurement with $d^4$ outcomes and wishes to encode these for Bob. At this point Alice and Bob share one maximally entangled state, as well as one qudit channel from Alice to Bob. Alice employs the QRAC-SE $2_{d^2} \xmapsto{p,1_{d^2}} (1_{d} , 1_{d})$ from Sec.~\ref{sec:QRAC-SE} to encode the $d^4$-inputs, which is two strings of two digits of base $d$. Bob decodes this information to guess $\psi_x$, with his choice of $x$.

For the qubit case Bob can guess $\psi_x$ with the fidelity that coincide with that of $2_4 \xmapsto{p,1_4} (1_{2} , 1_{2})$ QRAC-SE, that is $P_{\text{succ}}=0.728$, which is better than the trivial strategy $P_{\text{avg}}^{\text{trivial}} >0.625$ (see Table~\ref{Tab4}). Incidentally, the latter is equal to the lower bound \eqref{eqn20} for the NS-QRAC with a classical communication channel. It is not clear, whether such a strategy can offer an improvement for higher $d$.

\section{\label{sec:constrainedententanglement}A constrained entanglement scenario}

One may consider another modification to the NS-QRAC scenario. Now, instead of constrained classical communication and unbounded entanglement resource, we can consider restricted entanglement of size $d \times d$ and unconstrained classical communication. We call this modification the CNS-QRAC (Constrained-No-Signalling Quantum Random Access Code). In this context we ask the question: Can we compute, or at least find a reasonable upper bound, on the probability of success \eqref{N2}?

We show that this is indeed possible. We shall work in a `distant-laboratories' paradigm, which assumes no interaction between the parties and only pre-shared quantum correlations are allowed. We leave open the question of optimal transmission via general quantum no-signalling maps (see \cite{Piani}), where interaction is allowed but the no-signalling property is retained.

Consequently, the shared state together with the local data of Alice and Bob can at most be subject to a product of generalised measurements.
This can be simulated by a product of local unitaries (isometries), followed by measurements. Finally, any shared mixed state can be purified through an environment $E$ that can always be incorporated via some isometry for Alice or Bob. 
We thus arrive at the following:
 \begin{observation}
 \label{thm:main1}
 The quantum transmission of a CNS-QRAC box in a distant-laboratories paradigm can be effectively simulated by the scenario with a shared quantum state of dimension $d \times d$ dimension with local operations and classical communication from Alice to Bob. 
 In fact, it is enough for parties to use pure states only.
 \end{observation}
The above observation shows that a CNS-QRAC box can be viewed as a quantum channel with LOCC action. This means that all properties of the box have to obey the laws of quantum mechanics, in particular the monogamy relations for entanglement. 
However, the `monogamy' in this picture has a slightly different meaning than commonly adopted, see Figure~(\ref{cloningbip2}).
 \begin{figure}[h!]
  \includegraphics[width=1.0\linewidth]{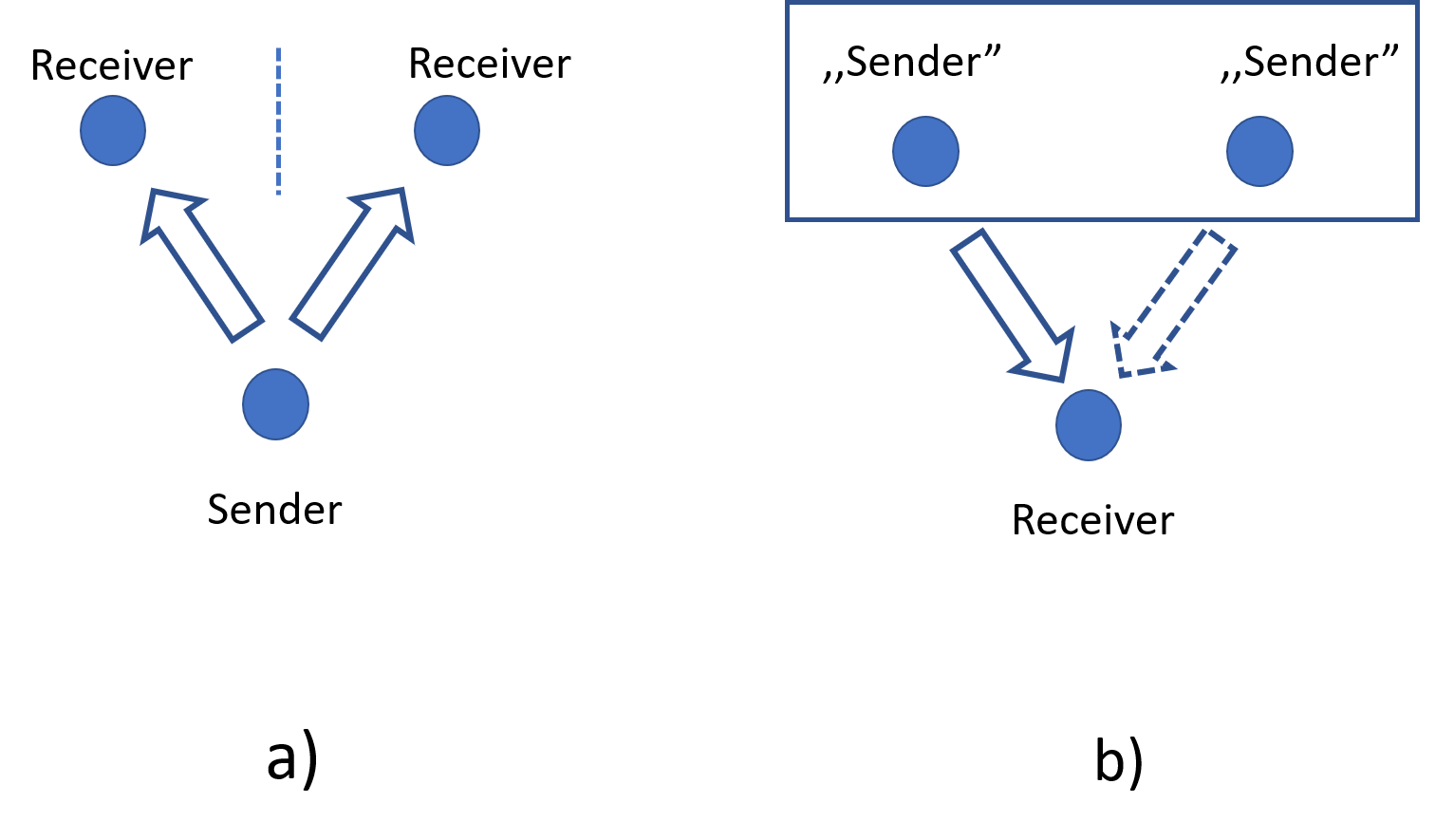}
  \caption{\label{cloningbip2} 
The usual quantum information scenario, where the monogamy plays a key role is depicted in panel a) where two spatially separated receivers are going to receive the same quantum information. 
The present scenario, b), is different -- we can interpret it as a situation where two senders are going to 
send different quantum information of a given dimension $d$ to a single receiver, the memory 
of which is restricted to a single system of the dimension $d$. However, the senders cooperate quantumly 
(their particles can interact quantumly) and only one of them is supposed to succeed at a time, or – in other 
words – only an alternative of the two successes is required (which is depicted
by the dashed contour of one of the arrows). Despite significant differences, the monogamy relation 
is known well to work in a) also bounds quantum transmission in b). The open question remains whether the latter bound 
can be saturated.
 }
 \end{figure}

Having turned the CNS-QRAC into a monogamy relations problem, we can ask what is the maximal bipartite entanglement which one can teleport from $A_1'A_2':A_1A_2$ to $A_1'A_2':B$, see Figure~\ref{cloningbip}. 
 \begin{figure}[h!]
  \includegraphics[width=1.0\linewidth]{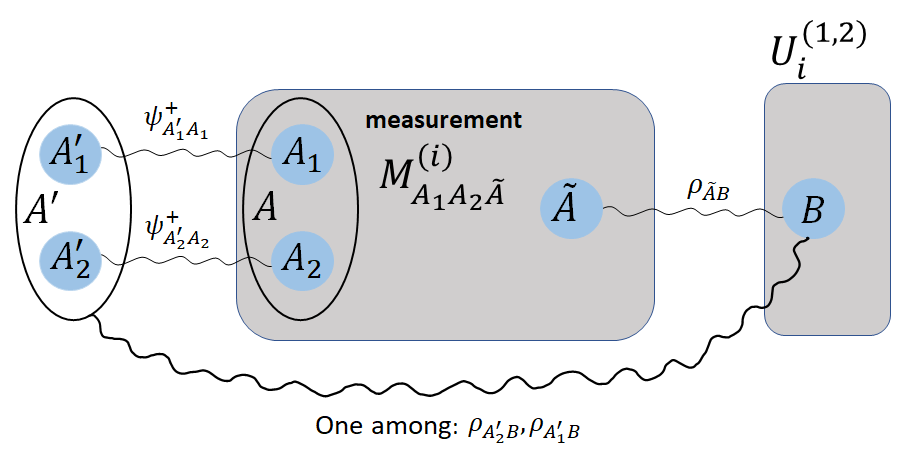}
  \caption{\label{cloningbip}The schematic description of multi-teleportation scheme with an unconstrained classical channel. In this situation, two parties, Alice and Bob, share a $d-$dimensional pure state $\rho_{\A B}$ and Alice wishes to send to Bob one of the states $A'_1$ or $A'_2$. Then Bob chooses which state of these two he wants to recover by applying respective unitaries $U_B^{(i,1)}, U_B^{(i,2)}$ depending on classical messages, indexed by $i$, sent by Alice. This scheme can naturally be extended to $N$ states on Alice's side. The presented scheme allows us to find an upper bound on $P_{\text{succ}}^{\text{QM}}$ in the CNS-QRAC game in terms of monogamy relations for entanglement. }
 \end{figure}
 The monogamy relations for entanglement were intensively studied in the context of universal quantum cloning machines~\cite{WernerCloning,kay2009optimal}, and we shall exploit these results here. We start by proving a result concerning a bipartite scenario, as depicted in Figure~\ref{cloningbip}.

\subsection{The quantum bound for CNS-QRAC} 
 We shall start this section with the following proposition: 
 \begin{prop}
 \label{bipartite}
 The probability of success $P_{\text{succ}}^{\text{QM}}$ in the symmetric CNS-QRAC scenario with two $d-$dimensional inputs $\psi_1,\psi_2$ satisfies the following bound:
 \begin{equation}
 P_{\mathrm{succ}}^{\mathrm{QM}}\leq \frac{d+1}{2d}.
 \end{equation}
 In the particular case of qubits we have $P_{\text{succ}}^{\text{QM}}\leq 3/4$.
 \end{prop}
 
 \begin{proof}
 For simplicity, we shall first present the proof for the case when the action CNS-QRAC comes from local operations on a shared maximally entangled state. Finally, we will show how the argument generalises for an arbitrary shared quantum state. Consequently, for the time being, let us assume that Alice and Bob share a maximally entangled state $\rho_{\A B}=|\psi^+\>\<\psi^+|_{\A B}$, where $| \psi^{+}_{\A B} \rangle = (1/\sqrt{d})\sum_{i=1}^{d} |i \rangle_{\A} |i \rangle_B$. Alice wishes to send one state, $A_1$ or $A_2$, to Bob (see figure~\ref{cloningbip}). Using the notation $A \vc A_1A_2$ and $A' \vc A'_1A'_2$, we write $\psi^+_{A'A}= \psi_{A'_1A_1}^+\otimes \psi_{A'_2A_2}^+$ to simplify the notation. Our goal is to estimate the entanglement fidelity $F$ of the whole process which, due to simulation argument from Observation~\ref{thm:main1} and discussion below, is equal to $P_{\text{succ}}^{\text{QM}}$, so: 
 \begin{align}
 \label{eq11}
&\!\!\! P_{\text{succ}}^{\text{QM}}=\frac{1}{2} (F_{A'_1B} + F_{A'_2B}) \\
&= \frac{1}{2} ( \tr[ \psi^{+}_{A'_1 B} \rho_{A'_1 B} ] + \tr[ \psi^{+}_{A'_2 B} \rho_{A'_2 B} ] ) \notag\\
&=\frac{1}{2}\sum_{i=1}^{K} \operatorname{tr}\left[\psi^{+}_{A'_1B} U_B^{(i,1)} M_{A\A}^{(i)}\left( \psi^+_{A'A} \otimes \rho_{\A B}\right)U_B^{(i,1)\dagger}\right] \notag\\
&\quad + \frac{1}{2}\sum_{i=1}^{K} \operatorname{tr}\left[\psi^{+}_{A'_2B}U_B^{(i,2)} M_{A\A}^{(i)}\left( \psi^+_{A'A} \otimes \rho_{\A B}\right)U_B^{(i,2)\dagger}\right]. \notag
\end{align}
Here K as before is the index for the POVM element that is communicated by Alice to Bob who then conditions the choice of his local unitary using this. Since we have unbounded classical communication CNS-QRAC, K can be arbitrary. For clarity let us focus on the first term in the last line of~\eqref{eq11}, since the second one is analogous. Using property \eqref{eq3} we can follow the same line of argumentation as in the proof of Proposition~\ref{coinstrainedF} getting
\begin{equation}
\begin{split}
&\sum_{i=1}^{K} \operatorname{tr}\left[\psi^{+}_{A'_1B} U_B^{(i,1)} M_{A\A}^{(i)} \left( \psi^+_{A'A} \otimes \rho_{\A B}\right)U_B^{(i,1)\dagger}\right] \\
&=\frac{1}{d^3}\sum_{i=1}^{K} \operatorname{tr}\left[( U_B^{(i,1)\dagger}\psi^{+}_{A'_1B} U_B^{(i,1)} ) M_{A'B}^{(i)t}\right]\\
&= \frac{1}{d^3}\sum_{i=1}^{K} \operatorname{tr}\left[\psi^{+,(i)}_{A'_1B} M_{A'B}^{(i)t}\right],
\end{split}
\end{equation}
where $\psi^{+,(i)}_{A'_1B}=U_B^{(i,1)\dagger}\psi^{+}_{A'_1B} U_B^{(i,1)}$, and $t$ denotes transposition.
Thus, we arrive at the following expression:
\begin{multline}
\label{someent}
P_{\text{succ}}^{\text{QM}}=\frac{1}{2d^3}\Bigg(\sum_{i=1}^{K}\operatorname{tr}\left[\psi^{+,(i)}_{A'_1B} M_{A'B}^{(i)t}\right] + \\
+ \sum_{i=1}^{K}\operatorname{tr}\left[\tilde{\psi}^{+,(i)}_{A'_2B} M_{A'B}^{(i)t}\right]\Bigg).
\end{multline}
Introducing tripartite states $\rho^{A'B}_{i}=\frac{M_{A'B}^{(i)t}}{\operatorname{tr}(M_{A'B}^{(i)t})}$
with the factors $\alpha_i=\operatorname{tr}(M_{A'B}^{(i)t})$ satisfying normalisation constrain $\sum_i \alpha_i=d^3$, we rewrite equation~\eqref{someent} as follows:
\begin{align}
\label{alphaeqn}
P_{\text{succ}}^{\text{QM}}&=\frac{1}{2d^3} \sum_{i=1}^{K}\alpha_i\left[\operatorname{tr}(\rho^{{A'B}}_{i} \psi^{+,(i)}_{A'_1B}) +\operatorname{tr}(\rho^{{A'B}}_{i}\tilde{\psi}^{+,(i)}_{A'_2B})\right] \notag \\
&\leq \frac{1}{2d^3} \sum_{i=1}^{K}\alpha_i\left[F_{\max}(\rho^{A'_1 B}_i) +F_{\max}(\rho^{A'_2B}_i) \right],
\end{align}
where for $j=1,2$ we define 
\begin{equation}
F_{\max}(\rho_i^{A'_jB}):=\max_{\psi^{+,(i)}_{A'_jB}} \operatorname{tr}(\rho_i^{A'_jB}\psi^{+,(i)}_{A'_jB}).
\end{equation}
Taking $F_{\max}(\rho^{A'_j B})=\max_i F_{\max}(\rho_i^{A'_jB})$, for $j=1,2$, we simplify expression~\eqref{alphaeqn} to
\begin{equation}
P_{\text{succ}}^{\text{QM}}\leq \frac{1}{2}\left[F_{\max}(\rho^{A'_1 B}) + F_{\max}(\rho^{A'_2B}) \right].
\end{equation}
It can be shown that for any tripartite quantum state $\rho^{A'_1 A'_2 B}$ the sum of the fidelities in the above formula must be strictly smaller than two\footnote{Let us consider a tripartite state $\rho_{ABC}$ with the respective marginals $\rho_{AB},\rho_{AC}$. Now assuming that the overlaps of the marginals with respective maximally entangled sates are equal 1 means that  $\rho_{AB}$ and $\rho_{AC}$ are maximally entangled and pure. From this fact it follows  that the state $\rho_{AB}$ is  product with the system $C$. The same argumentation holds for the state $\rho_{AC}$. The latter however, contradicts that there is no entanglement between $A$ and $C$. From this we conclude that indeed sum of the two fidelities must be strictly smaller than 2.}. This is just a manifestation of the famous quantum entanglement monogamy phenomenon --- the particle $B$ cannot be maximally entangled with $A'_1$ and $A'_2$ at the same time, or equivalently, the two fidelities can not be equal to unity at the same time. In the symmetric case, when all the entanglement fidelities should be equal we can use the result from~\cite{WernerCloning}. Namely, for a universal quantum cloning machine producing $N_2$ clones from $N_1$ input states ($N_1\rightarrow N_2$) the average fidelities of outputs are
\begin{equation}
\label{werner}
f=\frac{N_1}{N_2}+\frac{(N_2-N_1)(N_1+1)}{N_2(N_1+d)}.
\end{equation}
In our case we plug $N_1=1,N_2=2$ obtaining
\begin{equation}
f_{\max}(\rho^{A'_1B})=f_{\max}(\rho^{A'_2B})=\frac{1}{2}+\frac{1}{d+1},
\end{equation}
so we produce the following equality
\begin{equation}
 \frac{1}{2} \big(f_{\max}(\rho^{A'_1B})+f_{\max}(\rho^{A'_2B})\big)= \frac{d+3}{2(d+1)}. 
\end{equation}
 This, together with linear relation between the transmission fidelity $f$ and the entanglement fidelity $F$ from~\cite{HorodeckiMPRFidelity}, which reads $f=\frac{dF+1}{d+1}$, we can establish an upper bound for entanglement fidelity in the bipartite scenario:
\begin{equation}
P_{\text{succ}}^{\text{QM}}\leq \frac{1}{2}\left(F_{\max}(\rho^{A'_1 B}) + F_{\max}(\rho^{A'_2 B})\right)= \frac{d+1}{2d}.
\end{equation}
In the qubit case, when $d=2$ we obtain the upper bound $3/4$ from the statement.
This finishes the proof.
\end{proof}
 Notice that the argumentation presented in the proof of Proposition~\ref{bipartite} holds for any shared pure state $\rho_{\widetilde{A}B}=|\psi\>\<\psi|_{\widetilde{A}B}$ acting on the $(d'\cdot d)$-dimensional space, where
 \begin{equation}
|\psi\>_{\widetilde{A}B}=\sqrt{d}\left(\sqrt{\sigma_{\widetilde{A}}}\otimes \mathbf{1}_{B}\right)|\psi^+\>_{\widetilde{A}B}.
 \end{equation}
 Then, expression~\eqref{someent} reads 
\begin{multline}
\label{F_general}
P_{\text{succ}}^{\text{QM}}=\frac{1}{2d^3}\Bigg(\sum_{i=1}^{K}\operatorname{tr}\left[\psi^{+,(i)}_{A'_1B} \widetilde{M}_{A'B}^{(i)t}\right]+\\
+ \sum_{i=1}^{K}\operatorname{tr}\left[\widetilde{\psi}^{+,(i)}_{A'_2B} \widetilde{M}_{A'B}^{(i)t}\right]\Bigg),
\end{multline}
with $\sum_{i=1}^{k} \widetilde{M}_{A'B}^{(i)t}= \mathbf{1}_{A'}\otimes \sigma_{B}$, where $A'=(A'_1 A'_2)$. Note that the result for an 
arbitrary pure state extends by convexity to any $(d \times d)$-dimensional 
state $\rho_{\widetilde{A}B}$. The argument that the latter can be chosen to be $d \times d$ is straightforward: In fact the $\tilde{d} \times \tilde{d} $ case can easily be reduced to $\tilde{d} \times d$ due to the final fidelity with $\psi^{+}_{A'_i B}$ state, which is unitarily equivalent to a $d \times d$ state. Note that the latter fact means that in general the operators $M^{(i)}_{A\widetilde{A}}$ and $M^{(i)}_{A'B}$ differ. Namely, the matrix representation of the latter is the one of the former but constrained to the $d$-dimensional subspace on the subsystem $B$. However, since they are all positive, this can only decrease the probability of success, 
so we may define all of them on the system $B$ with the dimension $\tilde{d} =d$ and 
by similar argument the shared state $\rho_{\widetilde{A}B}$ can be 
constrained analogously.

\subsection{Multi-party and asymmetric scenarios}
 
 The CNS-QRAC game can be extended to the case of $N$ pure inputs $\psi_1,\ldots,\psi_N$, which are picked with some given probabilities $\{p_i\}_{i=1}^N$. The probability of success then reads
 \begin{align}
 \label{NN}
 P_{\text{succ}}^{\text{QM}} = \sum_{i=1}^N p_i \, F(\psi_i,\rho_i).
 \end{align}
 
 The problem reduces to finding the 
 set of $N$ fidelities in equation~\eqref{NN} by considering the problem of asymmetric universal quantum cloning machine producing $N$ clones from one input state.
 
In the symmetric case, i.e. when $p_i = 1/N$, we can again utilise the results from \cite{WernerCloning} to derive an explicit bound for $P_{\mathrm{succ}}^{\mathrm{QM}}$.

\begin{prop}
\label{multipartite}
 The probability of success $P_{\mathrm{succ}}^{\mathrm{QM}}$ in the symmetric
 CNS-QRAC scenario with $N$ inputs $\psi_1,\ldots,\psi_N$ of dimension $d$ each, satisfies the following bound:
 \begin{equation}\label{PdN}
 P_{\mathrm{succ}}^{\mathrm{QM}}\leq \frac{N+d-1}{dN}.
 \end{equation}
\end{prop}

\begin{proof}
Suppose that Alice has at her disposal $N$ maximally entangled states $\psi^+_{A'_1A_1},\psi^{+}_{A'_2A_2},\ldots, \psi^+_{A'_NA_N}$. Then the goal of Bob is to recover one of the states $\rho_{A'_iB}$, for $i=1,\ldots,N$. In this situation we can again use the argumentation based on the monogamy relations, this time for more than three parties. Indeed, plugging $N_1=1$ and $N_2 \mapsto N$ into~\eqref{werner}, we have
\begin{equation}
\frac{1}{N}\sum_{i=1}^N f_{\max}(\rho^{A'_iB})=\frac{1}{N}+\frac{2(N-1)}{N(d+1)}.
\end{equation}
Using the relation connection transmission fidelity and entanglement fidelity we end up with the following bound on $P_{\text{succ}}^{\text{QM}}$:
\begin{equation}
P_{\text{succ}}^{\text{QM}}\leq \frac{1}{N}\sum_{i=1}^N F_{\max}(\rho^{A'_iB})=\frac{N+d-1}{dN}.
\end{equation}
This finishes the proof.
\end{proof}

In order to provide a bound $P_{\text{succ}}^{\text{QM}}$ in the general case with a non-uniform probabilistic distribution of the input states one can employ the results from~\cite{kay2009optimal} connecting the entanglement fidelities $F_{\max}(\rho^{A'_iB})$:
 \begin{align}
 &\sum^{N}_{i=1} F_{\max}(\rho^{A'_iB}) \\
 & \qquad\quad \leq \frac{d-1}{d} + \frac{1}{N+d-1}\bigg(\sum^{N}_{i=1} \sqrt{ F_{\max}(\rho^{A'_iB})} \bigg)^2. \notag
\end{align}
With the notation $x_i \vc \sqrt{F_{\max}(\rho^{A'_iB})}$, the quantum bound for $P_\mathrm{succ}$ equals to
\begin{align}\label{PQx}
 \sum_{i=1}^N p_i \, x_i^2.
\end{align}
Hence, one has to maximise the quantity \eqref{PQx} over $x_i \in [0,1]$, under the constraint
\begin{equation}\label{constx}
 \sum^{N}_{i=1} x_i^2 \leq \frac{d-1}{d} + \frac{1}{N+d-1}\bigg(\sum^{N}_{i=1} x_i \bigg)^2.
\end{equation}
Eq. \eqref{constx} determines the interior of a rotated $N$-ellipsoid in $\mathbb{R}^N$. Because the function \eqref{PQx} is increasing in all variables $x_i$ it attains the maximum on the boundary. Consequently, it is sufficient to seek its maximal value on the hypersurface in $[0,1]^N$ determined by the equation
\begin{equation}\label{hyperx}
 \sum^{N}_{i=1} x_i^2 = \frac{d-1}{d} + \frac{1}{N+d-1}\bigg(\sum^{N}_{i=1} x_i \bigg)^2.
\end{equation}
In particular, for $N=2$ with $p_1 = 1- p_2 = p$ one finds the bound
\begin{align}
P_{\text{succ}}^{\text{QM}} \leq \frac{1}{2} \left(1 + \sqrt{1+ \frac{4 \left(d^2-1\right) (p-1) p}{d^2}} \, \right). 
\end{align}

\section{\label{sec:outlook}Discussion and outlook}

We studied two instances of 
random access codes using quantum information. The first one 
involved remote access to one of the two given quantum states {\it via} an
NS-QRAC box implemented quantumly in the `distant labs' paradigm.  We considered (see Sec. \ref{sec:NSQRACwithQRACSE}) a variation where a constrained quantum channel is used as contrasted to a constrained classical channel used in \cite{grudka2015nonsignaling}  . This, in a way, is the most natural quantum version of the $2\rightarrow 1$ classical RAC problem, because we have remote access to one of two qubits transmitted over a qubit channel. In this case, we found a lower bound for the probability of success $P_{\text{succ}}^\text{QM}\geq 0.728$.

We also considered another modification --- the CNS-QRAC (Sec. \ref{sec:constrainedententanglement}), where we find that the trade-off for information transmission 
corresponds to a typical monogamy relation. In this case, we provided a reasonable upper bound for the probability of success \eqref{PdN} for a general CNS-QRAC with $N$ input states of dimension $d$. 
An interesting aspect of this scenario is that here the transmission of quantum information does not involve a single sender and two receivers, as it is the case in the standard quantum channel capacity restrictions based on quantum cloning (see for example \cite{PhysRevA.57.2368,CubittT}).
Instead, the transmission goes from the composite system of two `senders' who cooperate quantumly to transfer quantum information to a single receiver.
Since the `senders' are required to transmit 
different quantum information, which is also 
supposed to come 
as an {\it alternative} rather than jointly, there seems to 
be no {\it a priori} reasons why the cloning bound should be 
obeyed. Nevertheless, it turns out to apply in such a scenario as well.

An interesting open problem is to extend this analysis to 
the quantumly simulable NS boxes, where the two 
parties may interact (see \cite{Piani}, \cite{Schmid_2020}). Clearly, when the labs are far apart, such boxes are super-quantum. In fact --- as shown in a recent paper \cite{Schmid_2021} --- its subclasses with classical inputs are even interconvertible with PR boxes with the help of shared entanglement and local operations. Hence, at the intuitive level, it is possible that the corresponding CNS-QRAC might allow both (all) fidelities to be perfect, but this conjecture would need further investigation. 

The second instance of random access codes, and to some extent a complementary scenario, has been introduced here to analyse the power of quantum entanglement when aiding quantum random access coding. To this end, we have defined and studied 
quantum random access codes with shared entanglement and a quantum channel. An interesting aspect of this problem occurs for the class of QRAC-SE $2_{d^2} \xmapsto{p,1_{d^2}} (1_{d} , 1_{d})$ problems, as the encoding by Alice depends on the existence of generalised Gray codes. This should be compared with the problem of QRAC-SR \cite{ambainis2008quantum}, where Alice's encoding depends on finding some form of symmetric quantum states in the Bloch sphere. The presented explicit protocols provide lower bounds for the probabilities of success. It is an open problem to find the relevant upper bounds, perhaps using numerical methods similar to the techniques involved in finding the upper bounds in \cite{tavakoli2021correlations}. Lastly, we provided a proof of concept for extending the QRAC-SE to $f-$QRAC-SE over Boolean functions, similar to the studies of $f-$QRAC in \cite{doriguello2021quantum}, which may inspire an interesting line of future research.

\section*{Acknowledgements}
We thank the comments of both reviwers which helped improve the presentation of this work. N.S. and P.H. acknowledge support by the Foundation for Polish Science (IRAP project, ICTQT, contract no. MAB/2018/5, co-financed by EU within Smart Growth Operational Programme). N.S acknowledges the partial support of an NSERC grant as well as acknowledges that research at Perimeter Institute is supported in part by the Government of Canada. M.S. is supported by the NCN grant Sonatina 2 UMO-2018/28/C/ST2/00004. M.E. acknowledges the support of the Foundation for Polish Science under the Team-Net Project no. POIR.04.04.00-00-17C1/18-00.

\bibliography{bib.bib}
\appendix

\end{document}